\documentclass[aps,reprint,showpacs,onecolumn,tightenlines]{revtex4-1}
\usepackage[utf8]{inputenc}
\usepackage{float}
\usepackage{times}
\usepackage{amsmath, amsthm, amssymb, bm}
\usepackage{mathrsfs}
\usepackage{extarrows}
\usepackage{pifont}
\usepackage{diagbox}
\usepackage{xcolor}
\usepackage{multirow}
\usepackage{graphicx}
\usepackage{hyperref}
\hypersetup{  colorlinks=true, linkcolor=blue, citecolor=red, urlcolor=blue }

\usepackage{cleveref}
\usepackage[normalem]{ulem}
\usepackage{braket}
\usepackage{pgfplots}
\pgfplotsset{compat=1.18}
\usepackage[titletoc,title]{appendix}
\makeatletter
\let\newfloat\newfloat@ltx
\makeatother
\usepackage{algorithm}
\usepackage{algorithmicx}
\usepackage{algpseudocode}

\makeatletter
\renewcommand*{\ALG@name}{Algorithm }
\makeatother

\algblock{Input}{EndInput}
\algnotext{EndInput}
\algblock{Output}{EndOutput}
\algnotext{EndOutput}

\definecolor{ultramarine}{rgb}{0.07, 0.04, 0.56}

\newtheorem{theorem}{Theorem}
\newtheorem{lemma}{Lemma}

\DeclareMathOperator{\sgn}{sgn}

\DeclareMathOperator{\Cl}{Cl}

\newcommand{\abs}[1]{\left| #1 \right|}
\newcommand{\vabs}[1]{\left\| #1 \right\|}
\newcommand{\pbra}[1]{\left( #1 \right)}
\newcommand{\cbra}[1]{\left\{ #1 \right\}}
\newcommand{\sbra}[1]{\left[ #1 \right]}

\newcommand{\floor}[1]{\left\lfloor #1\right \rfloor }

\newcommand{\Cbb}{\mathbb{C}}
\newcommand{\Ebb}{\mathbb{E}}

\newcommand{\Ibb}{\mathbb{I}}
\newcommand{\Mbb}{\mathbb{M}}

\newcommand{\Rbb}{\mathbb{R}}

\newcommand{\Bcal}{\mathcal{B}}

\newcommand{\Ecal}{\mathcal{E}}

\newcommand{\Hcal}{\mathcal{H}}
\newcommand{\Ical}{\mathcal{I}}

\newcommand{\Lcal}{\mathcal{L}}
\newcommand{\Mcal}{\mathcal{M}}
\newcommand{\Ocal}{\mathcal{O}}
\newcommand{\Pcal}{\mathcal{P}}
\newcommand{\Rcal}{\mathcal{R}}

\newcommand{\Xcal}{\mathcal{X}}
\newcommand{\Ucal}{\mathcal{U}}

\newcommand{\Dmat}{\bm{\mathrm D}}
\newcommand{\Gmat}{\bm{\mathrm G}}

\newcommand{\poly}{\mathrm{poly}}

\newcommand{\MedianOfMeans}{\mathrm{\mathbf{MedianOfMeans}}}
\newcommand{\pf}{\mathrm{pf}}

\newcommand{\Ord}[1]{\mathcal{O}\left( #1 \right)}

\newcommand{\Var}[1]{\mathrm{Var} \left( #1 \right)}

\newcommand{\ketbra}[2]{\ket{#1}\!\bra{#2}}
\newcommand{\supket}[1]{|#1 \rangle\!\rangle}
\newcommand{\supbra}[1]{\langle\!\langle #1 |}
\newcommand{\supbraket}[1]{\langle\!\langle #1 \rangle\!\rangle}
\newcommand{\tr}[1]{\mathrm{Tr}\left( #1 \right)}

\begin{document}

\title{Error-mitigated fermionic classical shadows on noisy quantum devices}

\author{Bujiao Wu}
\email{bujiaowu@gmail.com}
\affiliation{Dahlem Center for Complex Quantum Systems, Freie Universit\"{a}t Berlin, 14195 Berlin, Germany}
\affiliation{{Shenzhen Institute for Quantum Science and Engineering,
Southern University of Science and Technology, Shenzhen 518055, China}}
\affiliation{Center on Frontiers of Computing Studies, Peking University, Beijing 100871, China}
\author{Dax Enshan Koh}
\email{dax\textunderscore koh@ihpc.a-star.edu.sg}
\affiliation{Institute of High Performance Computing (IHPC), Agency for Science, Technology and Research (A*STAR), 1 Fusionopolis Way, \#16-16 Connexis, Singapore 138632, Republic of Singapore}

\begin{abstract}
Efficiently estimating fermionic Hamiltonian expectation values is vital for simulating various physical systems. Classical shadow (CS) algorithms offer a solution by reducing the number of quantum state copies needed, but noise in quantum devices poses challenges. We propose an error-mitigated CS algorithm assuming gate-independent, time-stationary, and Markovian (GTM) noise. For $n$-qubit systems, our algorithm, which employs the easily prepared initial state $|0^n\rangle\!\langle 0^n|$ assumed to be noiseless, efficiently estimates $k$-RDMs with $\widetilde{\mathcal O}(kn^k)$ state copies and $\widetilde{\mathcal O}(\sqrt{n})$ calibration measurements for GTM noise with constant fidelities. We show that our algorithm is robust against noise types like depolarizing, damping, and $X$-rotation noise with constant strengths, showing scalings akin to prior CS algorithms for fermions but with better noise resilience. Numerical simulations confirm our algorithm's efficacy in noisy settings, suggesting its viability for near-term quantum devices.
\\
\\
\textbf{keywords:}
classical shadow; fermionic system; error mitigation; randomized benchmarking
\end{abstract}
\maketitle

\section*{Introduction}
Assessing the properties of interacting fermionic systems constitutes one of the core tasks of modern physics, a task that has a wealth of applications in quantum chemistry \cite{CramerQuantumChemistry}, condensed matter physics \cite{Sachdev}, 
and materials science \cite{QuantumMaterials}. Notions of \textit{quantum simulation} offer an alternative route to studying this important class of systems. 
In \textit{analog simulation}, one prepares 
the system of interest under highly controlled conditions. However, any such effort makes sense only if one has sufficiently powerful \textit{readout techniques} available that allow one to estimate properties. In fact, the read-out step constitutes a core bottleneck in many schemes for quantum simulation. 

Fortunately, for natural fermionic systems, one often does not need to learn the full unknown quantum state; trying to do so regardless would be highly impractical, as the resources required for a full tomographic recovery would scale exponentially with the size of the system. Instead, what is commonly needed are the so-called \textit{$k$-particle reduced density matrices}, abbreviated as $k$-RDMs. 
These are expectation values of polynomials 
of fermionic operators
of the $2k$-th degree. Naturally, the expectation value of any interaction fermionic Hamiltonian can be estimated using $2$-RDMs only
\cite{schwerdtfeger2011testing,Peterson13More}. Indeed, the adaptive \textit{variational quantum algorithm} (VQE) \cite{Variational} also utilizes up to $4$-RDMs to simulate many-body interactions in the ground and excited state~\cite{Parrish19Quantum,Takeshita20Increasing}. That is to say, meaningful methods of
read-out often focus on estimating such 
fermionic reduced density matrices.

On the highest level, several approaches can be pursued when dealing with fermionic operators. One of those---and the one 
followed here---is to treat the fermionic system basically as
a collection of spins. {Then given spin
Hamiltonians $\cbra{H_i}_{i=1}^m$ and an unknown quantum state $\rho$, where $m= \Ord{\poly(n)}$, the \textit{classical shadow} (CS) algorithm or its variants~\cite{huang2020predicting,hadfield2022measurements,Huang21efficient,Wu23overlapped,
hadfield2021adaptive, hu2021classical, archarya2021shadow, bu2022classical,grier2022sample, ippoliti2023classical,zhou2023performance,Garcia2021Quantum,zhou2023hybrid} in qubit systems are among the most promising ways to calculate the expectations $\tr{\rho H_i}$, with the representation of the Hamiltonians $\cbra{H_i}_{i = 1}^m$ in the Pauli basis, which
invokes a \textit{fermion-to-spin mapping}
such as the \textit{Jordan–Wigner} \cite{jordan1993paulische,nielsen2005fermionic} or \textit{Bravyi-Kitaev
encodings}~\cite{bravyi2002fermionic,tranter2015b}.} We define the classical shadow channel as $\Mcal$, which involves operating the unitary channel $\Ucal$ uniformly randomly sampled from the Clifford group before measurements in the $Z$-basis measurements and classical postprocessing operations on the measurement outcomes. By performing the inverse of the classical shadow channel $\Mcal^{-1}$ on the resulting state after performing $\Mcal$ on the initial state $\rho$, we obtain the classical shadow representation $\hat{\rho}$ of the quantum state $\rho$, allowing for the calculation of the expected values of observables $\cbra{H_i}_{i=1}^m$ with respect to $\rho$ using classical methods.

While the classical shadow algorithm requires exponentially many copies even for some local interacting fermions due to the inefficient representation in the qubit system,
recently, several classical shadow algorithms for fermionic systems without encoding of the Hamiltonians have been proposed~\cite{Zhao21Fermionic,low2022classical,wan2022matchgate}.  
Zhao, Rubin, and Miyake~\cite{Zhao21Fermionic} utilize the generalized CS method~\cite{huang2020predicting} for fermionic systems, and proposed an algorithm that requires $\Ord{{n\choose k}k^{3/2} (\log n)/\varepsilon^2}$ 
copies for the unknown quantum states to output all the elements of a $k$-RDM.
Low~\cite{low2022classical} proves that all elements of the $k$-RDM can be estimated with  ${{\eta\choose k}(1-\frac{\eta - k}{n})^k\frac{1+n}{1+n-k}/\varepsilon^2}$ number of copies of the quantum state, where $\eta$ is the number of particles and $n$ is the number of modes. These fermionic shadow estimation methods (along with the generic classical shadow formalism) do not account for noise in the system, which is an inevitability {in real physical systems}.

Since we are still in the noisy intermediate-scale quantum (NISQ) era, current quantum simulators are heavily affected by noise; hence, any characterization technique needs to be robust for these simulators to be useful. For qubit systems, robust shadow estimation was developed \cite{chen2021robust, koh2022classical} where Chen et al~\cite{chen2021robust} use techniques from randomized benchmarking to mitigate the effect of gate-independent time-stationary Markovian (GTM) noise channels on the procedure. Jnane et al.~\cite{jnane2023quantum} proposed error-mitigated classical shadow with probabilistic error cancellation.

Utilizing the robust shadow estimation scheme and taking inspiration from the fermionic shadow estimation of Zhao et al~\cite{Zhao21Fermionic, wan2022matchgate}, we present an error-mitigated shadow estimation scheme for fermionic systems and demonstrate its feasibility for realistic noise channels.
{Note that akin to the fermionic CS approaches proposed in Refs.~\cite{Zhao21Fermionic,wan2022matchgate}, our error-mitigated CS method circumvents the need to encode the Hamiltonian using the qubit representation.}

We sample our unitaries $\Ucal_Q$ from the matchgate group~\cite{valiant2002Expressiveness}, a natural choice for our protocol as there is a one-to-one correspondence between two-qubit matchgates and free-fermionic evolution~\cite{knill2001fermionic,Terhal2002Classical}. We therefore design the classical postprocessing operations by leveraging the irreducible representation of the matchgate group. We successfully introduce an unbiased estimator $\widehat{\Mcal}$ for the noisy classical shadow channel $\widetilde{\Mcal}$, where we require an additional calibration protocol to generate the estimator $\widehat{\Mcal}$ with the assumption that the computational basis state $\ketbra{\bm 0}{\bm 0}$ can be prepared noiselessly. 
    Additionally, we demonstrate the efficacy of our protocol under conditions of constant noise strength by evaluating its performance across various common noise channels: depolarizing noise, {generalized} amplitude damping, $X$-rotation, and Gaussian unitary noise. The number of samples required for the estimation process of our protocol is {in} the same order as the noise-free matchgate classical shadow scheme~\cite{Zhao21Fermionic, wan2022matchgate}.

We determine the effectiveness of our protocol with the above noise models by calculating the expectations of all elements of the $k$-particle reduced density matrix ($k$-RDM) when the noise strength is constant. The number of samples required for estimation, in this case, is $\Ord{kn^{k}\ln (n/\delta_e)/\varepsilon_e^2}$ and for calibration is $\Ord{\sqrt{n}\ln n\ln (1/\delta_c)/\varepsilon_c^2}$ with error $\varepsilon_e + \varepsilon_c$ and success probability $(1-\delta_e)(1-\delta_c)$.

We have extended the analysis of our error-mitigated fermionic shadow channel estimation to more general physical quantities inspired by the fermionic shadow analysis of Wan et al.~\cite{wan2022matchgate}, with more details in Supplementary Notes 5, 9. 
We list distinct classical shadow approaches in both noiseless and noisy qubit and fermionic systems in Table \ref{tab:summary_referendes}. Our error-mitigated fermionic classical shadow technique constitutes an extension of the work by Chen et al.~\cite{chen2021robust}, accommodating scenarios where the gate-set lacks (1) 3-design properties~\cite{zhu2016clifford} and (2) the applicability of the randomized benchmarking scheme developed by Helsen et al.~\cite{helsen2019new}.

We tested the accuracy and efficacy of our protocol by performing numerical experiments to estimate $\tr{\rho \widetilde{\gamma}_S}$ (where $\widetilde{\gamma}_S = U_Q^{\dagger}\gamma_S U_Q$, where $\gamma_S$ is the product of $\abs{S}$ Majorana operators and plays a crucial role in computing $k$-RDMs) on a noisy quantum device subjected to various types of gate noise such as depolarizing, {generalized} amplitude damping, $X$-rotation, and Gaussian unitaries. Our numerical investigations confirm the potential of our methods in real-world experimental scenarios.

\begin{table}[]
    \centering
        \caption{Enumeration of the classical shadow protocols in noiseless and noisy settings, for qubit and fermionic systems respectively.}
    \begin{tabular}{p{1.5cm}|p{3.5cm}|p{3.5cm}}
    \hline
        & \centering 
        Clifford-based shadows
         & Fermionic shadows
         \\ \hline
\centering Noiseless & \raggedright Huang~et~al.~\cite{huang2020predicting} & Zhao~et~al.~\cite{Zhao21Fermionic}; Low~\cite{low2022classical}; \newline Wan~et~al.~\cite{wan2022matchgate} \\ \hline
        \centering Noisy & \raggedright Chen~et~al.~\cite{chen2021robust}; Koh~and~Grewal \cite{koh2022classical} & This work
        \\
        \hline
    \end{tabular}
   \label{tab:summary_referendes}
\end{table}

\section*{Results}

\subsection*{Basic notations and concepts}
\label{sec:preliminary}

Here we give the basic notations and concepts that will be used throughout this work.

\textbf{Basic notations.}
The symbols $X$, $Y$, and  $Z$ denote the Pauli $X$, $Y$, and $Z$ operators respectively.  The operator $R_X(\theta) = \exp\pbra{-i\frac{\theta}{2} X}$ denotes the rotation operator around the $x$-axis. A $Z$-basis measurement is performed with respect to the basis of eigenstates of the Pauli-$Z$ operator. We utilize $\Ibb$ to represent the identity operator on the full system. The set of linear operators on a vector space $\Hcal$ is denoted as $\Lcal(\Hcal)$.
We utilize the symbol $\widetilde {\mathcal{O}}$ to omit the logarithmic terms.

\textbf{Superoperator.}
We denote the superoperator representation of a linear operator $O\in \Lcal(\Hcal)$ as $\supket{O}:= O/\sqrt{\tr{OO^{\dagger}}}$ and the scaled Hilbert-Schmidt inner product between linear operators as $\supbraket{O|R} = \tr{O^\dagger R}/\sqrt{\tr{OO^\dagger}\tr{ R R^\dagger}}$. The action of a channel $\Ecal \in \Lcal(\Lcal (\Hcal))$ operating on a linear operator $O \in \Lcal(\Hcal)$ can hence be written as 
$\Ecal\supket{O} = \Ecal(O)/\sqrt{\tr{OO^\dagger}}$. The channel representation of a measurement with respect to the computational basis can be represented as $\Xcal = \sum_{x\in\cbra{0,1}^n}\supket{x}\!\supbra{x}$. We denote the unitary channel corresponding to the unitary operator $U$ as $\Ucal(\cdot ) = U (\cdot )U^\dagger$.

\textbf{Majorana operator.} The Majorana operators
$\gamma_j$ for $1\leq j\leq 2n$  describes the fermionic system with
$\gamma_{j} = b_{(j+1)/2} + b_{(j+1)/2}^\dagger$ for odd $j$ and $\gamma_{j} = -i(b_{j/2}-b_{j/2}^\dagger)$ for even $j$, where $b_j$ and $b_j^\dagger$ are the annihilation and creation operators, respectively, associated with the $j$-th mode. Let
$\gamma_S$ be the product of the Majorana operators indexed by the subset $S$, denoted as $\gamma_S= \gamma_{l_1}\cdots \gamma_{l_{|S|}}$ for $\abs{S}>0$ and $\gamma_{\emptyset}= \Ibb$, where $S=\cbra{l_1,\ldots, l_{|S|}}$ and $l_1< l_2 < \ldots < l_{|S|}$. It can be shown that $\gamma_S$ forms the complete orthogonal basis for $\Lcal(\Hcal)$ for $S\subseteq [2n]$. 
Let $\Gamma_k := \cbra{\gamma_S|\abs{S} = k}$ be the subspace of $\gamma_S$ with cardinality $k$. We denote the even subspace as $\Gamma_{\mathrm{even}}:= \bigoplus_l \Gamma_{2l}$.
Also, we denote $\Pcal_k$ as the projector onto the subspace $\Gamma_k$, i.e.
\begin{equation}
    \Pcal_k = \sum_{S \in {{[2n]}\choose{k}}}\supket{\gamma_S}\supbra{\gamma_S},
\label{eq:projection_fermionic_system}
\end{equation}
where we have used the notation that for a set $A$ and an integer $k$, ${A\choose k} = \{T \subseteq A: |T| = k\}$ denotes the set of subsets of $A$ with cardinality $k$.

\textbf{Gaussian unitaries.}
Matchgates are in a one-to-one correspondence with the fermionic Gaussian unitaries and can serve as a qubit representation for these unitaries. We denote $\Mbb_n$ as the matchgate group, and write its elements $U_Q\in \Mbb_n$ in terms of rotation matrices $Q$ belonging to the orthogonal group Orth$(2n)$ (see Supplementary Note 1 for details)~\cite{knill2001fermionic,divincenzo2005fermionic,divincenzo2005fermionic}. 
Following Wan et al.'s study \cite{wan2022matchgate}, which demonstrated that the continuous matchgate group $\Mbb_n$ and the discrete subgroup $\Mbb_n \cap \Cl_n$ (where $\Cl_n$ represents the Clifford group) deliver equivalent performances for fermionic classical shadows, our findings remain applicable to both continuous and discrete matchgate circuits. Since $U_Q^\dagger \gamma_j U_Q = \sum_{k} Q_{jk} \gamma_k$, the matchgate $U_Q$ transforms  the product of Majorana operators $\gamma_S$ in the $\Gamma_{|S|}$ subspace as $U_Q^\dagger \gamma_S U_Q = \sum_{S'\in {[2n]\choose |S|}} \det(Q|_{SS'}) \gamma_{S'}$.

\textbf{$k$-particle reduced density matrices ($k$-RDM).}
We denote a $k$-RDM as $^{\bm k}\Dmat$, which can be obtained by tracing out all but $k$ particles. Here we denote it as a tensor with $2k$ indices,
\begin{align}
^{\bm k}\Dmat_{j_1,...,j_{k};l_1,...,l_k} = \tr{\rho b_{j_1}^{\dagger}\cdots b_{j_{k}}^{\dagger} b_{l_1}\cdots b_{l_k}},
\end{align}
where $j_i$ and $l_i$ are in $[n]$ for $i\in [k]$. The fermionic system can be equivalently described in the Majorana
basis, in which case a tensor can be rewritten as the linear combinations of $\tr{\rho\gamma_{S}}$, and $\abs{S} \leq 2k$. Hence
all $n^{2k}$ elements of the $k$-RDM can be obtained by calculating 
$\tr{\rho\gamma_S}$, for the scaling of $\Ocal\pbra{n^k}$ different $S$ with $\abs{S}\leq  2k$~\cite{Bonet20Nearly}.

\textbf{Pfaffian function.} The Pfaffian of a matrix $Q\in \Rbb^{2n\times 2n}$ is defined as
\begin{equation}
    \pf(Q) = \frac{2^n}{n!} \sum_{\sigma \in S_{2n}}
\sgn\pbra{\sigma} \prod_{i = 1}^n Q_{\sigma_{2i-1},\sigma_{2i}},
\end{equation}
which can be calculated in $O\pbra{n^3}$ time~\cite{wimmer2012algorithm}.

\textbf{Ideal fermionic shadow (Wan et al.~\cite{wan2022matchgate})}  
Given an unknown quantum state $\rho$, the classical shadow method applies a unitary $U_Q$ uniformly randomly sampled from matchgate group $\Mbb_n$, followed by measuring the generated state in the computational basis. With the measurement result $\supket{x}$, we can generate the classical representation $\hat{\rho}= \Mcal^{-1}\Ucal_{Q}^{\dagger}\supket{x}$ for the unknown quantum state $\rho$, where the channel $\Mcal$ describing the overall process is defined as

\begin{align}
\begin{aligned}
   \Mcal(\rho) &= \int_{Q} d\mu(Q)\sbra{\sum_{x\in \cbra{0,1}^n} \braket{x|U_Q\rho U_Q^\dagger|x} U_Q^\dagger \ketbra{x}{x}U_Q} \\
&= \sum_{k}\frac{{n\choose k}}{{2n\choose 2k}}\Pcal_{2k}(\rho). 
\label{eq:fermion_channel}
\end{aligned}
\end{align}

\begin{table*}[t]
    \centering
\caption{{Comparison between average noise fidelity $F_{\text{ave}}$, $Z$-basis average noise fidelity $F_{Z}$, and average noise fidelity in $\Gamma_{2k}$ subspace $\Bcal_k$.}}
   \begin{tabular}{c|c|c|c}
    \hline\hline
     Noise type& $\Lambda_{\text{d}}$ &
     $\Lambda_{\text{a}}$& $\Lambda_{\text{r}}$\\
    \hline
$\Lambda_{\text{avg}}$ & $1-p/2$ & $2/3 -(p_0+p_1)/6 + \sqrt{(1-p_0)(1-p_1)}/{3}$ & $(\cos (\theta) + 2)/{3}$
    \\
    \hline
         $F_{Z}$ & $1-p/2$ & $1 - (p_0 + p_1)/2$& $\cos^2 (\theta/2)$ 
        \\
          \hline
        $\Bcal_1$ &$1-p$ & $1 - (p_0 + p_1)$ &$\cos\theta$ \\ 
         \hline\hline
    \end{tabular}
    \label{tab:AveFidelityComp}
\end{table*}

\textbf{Noise assumptions.}
In this work, we assume that the noise is gate-independent, time-independent, and Markovian (a common assumption in randomized benchmarking (RB) abbreviated as the \textit{GTM noise assumption} \cite{flammia2020efficient}) and that the preparation noise for the easily prepared state $\ketbra{\bm 0}{\bm 0}$ is negligible. For the convenience of calculation, we utilize the left-hand side noisy representation for a noisy fermionic unitary $\widetilde{\Ucal}_{Q}:=  \Lambda \Ucal_{Q}$.
Here we define \textit{the average fidelity in $\Gamma_{2k}$ subspace} for noise channel $\Lambda$ as
\begin{align}
\Bcal_k := \frac{(-i)^k}{2^n{n \choose k}} \sum_{x} \sum_{S \in {[n]\choose k}} (-1)^{x_S} \tr{\ketbra{x}{x}\Lambda(\gamma_{D(S)})},
\label{eq:def_B_k}
\end{align}
where $D(S) = \cbra{2j-1, 2j|j \in S}, 0\leq k \leq n,x_S = \sum_{j\in S}x_j$. 
It is easy to check that $\Bcal_k = 1$ if $k=0$. 
With some calculations, we have $\Bcal_k = 1$ for the noise-free model where noise channel $\Lambda$ equals the identity.
In the following, we give the analysis of the simplified result for $\Bcal_k$ for several common noise models in the qubit system and fermionic system when $k>1$. See more details for the analysis in Supplementary Note 3.

(1) The depolarizing noise with channel representation $\Lambda_{\mathrm d}(A) = (1-p)A + p\tr{A}\frac{\Ibb}{2^n}$ for any $n$-qubit linear operator $A$, where the noise strength $p\in [0,1]$, and ${\Bcal_k}= 1-p$. 

(2) The {generalized} amplitude-damping noise with the Kraus representation
\begin{align}
\Lambda_{\mathrm a}(\cdot) = \sum_{
\substack{u, v\in \cbra{0,1}^n\\
u\ne v}
} E_{uv} (\cdot) E_{uv}^{\dagger} + E_0 (\cdot) E_0^{\dagger},
\label{eq:amp_noise}
\end{align}
where $E_{uv} = \sqrt{\bar p_{u}}\ketbra{v}{u}$ for $u\neq v \in \{0,1\}^n$ and   $E_0 = \sqrt{\Ibb - \sum_{
\substack{u,v\in \cbra{0,1}^n\\
u\ne v}
}E_{uv}^{\dagger}E_{uv}}$, where the probabilities $\bar p_u$ satisfy $(2^n-1)\bar p_{u}\leq 1$ for any $u\in \cbra{0,1}^n$. The average fidelity $\Bcal_{k} = 1 - \sum_{u\in \cbra{0,1}^n} \bar{p}_u$ if $k\ne 0$. We let $\sum_{u\in \cbra{0,1}^n} \bar{p}_u$ denote the noise strength.

(3) The $X$-rotation noise with the channel representation
\begin{align}
\Lambda_{\mathrm{r}}(\cdot) = R_{X}(\bm \theta)(\cdot)R_{X}(-\bm \theta)
\label{eq:Rx_noise}
\end{align}
 where $R_{X}(\bm \theta) = \exp(-i\sum_{l=1}^n \theta_l X_l/2)$, where the noise strengths $\theta_l$ are some real numbers. By some calculations, we have $\Bcal_k = {n\choose k}^{-1}\sum_{S\in {[n]\choose k}}\prod_{l\in S}\cos\theta_l$. Hence $\min_{l}\cos^k \theta_l \leq \Bcal_k\leq \max_{l}\cos^k \theta_l$.

(4) Noise that
is a Gaussian unitary~\cite{campbell2015decoherence,onuma2019classification}, where we assume that the noise has no coherence with the environment. This noise channel is denoted as 
\begin{align}
\Lambda_{\mathrm g}(\cdot) = U_Q (\cdot) U_Q^\dagger,
\end{align}
where $U_Q$ is a Gaussian unitary. 
By selecting the noise model to be $\Lambda_{\mathrm g}$, we get $\Bcal_k = {n\choose k}^{-1} \sum_{S,S'\in {[n]\choose k}} \det(Q|_{D(S),D(S')})$. 

 Note that for the noise models defined in (1--3), the average fidelity  $\Bcal_k\in [0,1]$ is close to one when the noise strengths are close to zero.

{For comparison, it is worth noting that the standard average noise fidelity \cite{Magesan2011scalable} $F_{\text{avg}} =\int_{\psi} d\psi \braket{\psi |\Lambda(\ket{\psi}\bra{\psi}) |\psi}$ where $\psi$ is drawn from the Haar measure, and the $Z$-basis average noise fidelity defined in Chen et al.~\cite{chen2021robust}, $F_{Z} = \frac{1}{2^n} \sum_{b} \supbraket{b |\Lambda|b}$ are not equivalent to $\Bcal_k$ under the same noise model. We present a comparison of these three quantities for $\Lambda_{\text{d}}, \Lambda_{\text{a}},\Lambda_{\text{r}}$ for a single qubit, as depicted in Table \ref{tab:AveFidelityComp}. They are closely aligned, with $\Bcal_1$ slightly smaller than $F_{\text{avg}}$ and $F_Z$. We give a more detailed analysis in Supplementary material X.}

 \begin{figure*}[t]
    \centering
\includegraphics[trim = 0mm 30mm 0mm 6mm, clip=true,width = 0.9\textwidth]{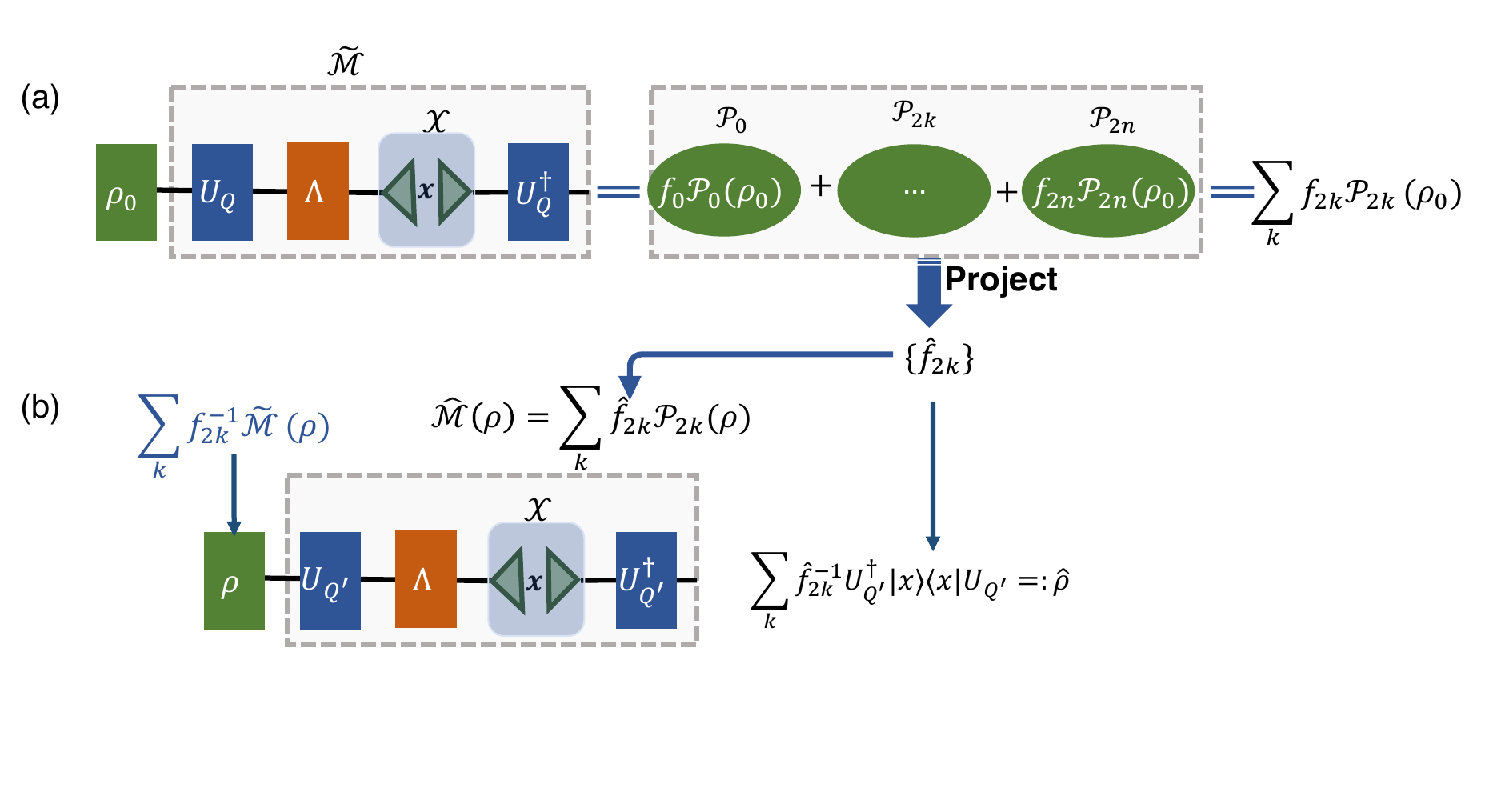}
    \caption{Schematic diagram of the error-mitigated matchgate classical shadow algorithm. (a) Protocol to learn the noisy classical shadow channel, where $\widetilde{\Mcal}=\Ucal_Q^{\dagger}\circ \Xcal \circ \Lambda \circ \Ucal_Q(\cdot) = \sum_{k = 0}^n f_{2k}\Pcal_{2k}(\cdot)$, where $\Ucal_Q$ is uniformly randomly sampled from the matchgate group. The estimation $\hat{f}_{2k}$ for $f_{2k}$ is obtained by projecting $\widetilde{\Mcal}$ to $\Pcal_{2k}$ subspace with input state $\rho_0=\ketbra{\bm 0}{\bm 0}$.  (b) The shadow estimation $\hat{\rho}$ for input state $\rho$ with the noisy quantum circuit and classical post-processing with the approximated noisy classical shadow channel in (a), where $\Ucal_{Q'}$ is uniformly randomly sampled from the matchgate group.}
    \label{fig:scheme}
\end{figure*}

\subsection*{Mitigation algorithm and error analysis}
\label{sec:mitigation_alg}

{Let $\widehat\Mcal:=\sum_{k=0}^n \hat{f}_{2k}\Pcal_{2k}$ be the estimator for the noisy shadow channel $\widetilde{\Mcal} = \sum_{k=0}^n f_{2k}\Pcal_{2k}$. In the Methods Section, we provide an explicit expression and efficiency for the calculation of $\hat{f}_{2k}$.  
Using the estimated noisy channel $\widehat{\Mcal}$, we can now obtain an estimate for $\cbra{\tr{\rho H_j}}_{j=1}^m$, where $\rho$ represents certain quantum states and $H_j$ denotes certain observables.} 
 {If $f_{2k} = 0$, the channel $\widetilde{\mathcal{M}}$ becomes non-invertible. Consequently, the effectiveness of the fermionic CS method diminishes, and we cannot retrieve $\text{tr}(\rho H)$ using it. In this study, we operate under the assumption that the noise is permissible and the fermionic CS channel is consistently invertible. Additionally, we provide a scenario in Supplementary Note 3 where the extreme noise channel occurs, i.e., $f_k = 0$. However, it is anticipated that such an extreme case will rarely occur.  Algorithm \ref{alg:mitigated_estimation_alg} demonstrates the method for mitigated estimation.}

\begin{algorithm}
\caption{Error-mitigated estimation for noisy fermionic classical shadows}
\label{alg:mitigated_estimation_alg}
\begin{algorithmic}[1]
 \Input
 { Quantum state $\rho$, observables $H_1,\ldots, H_m$, integers $N_c,K_c,N_e,K_e$.}
 \EndInput
\Output
{ $\hat{v}_i$ for $i\in [m]$.}
\EndOutput
\State $R_c := N_cK_c$;
\For{$j\leftarrow 1 $  to $R_c$}
\State Prepare state $\rho_0 = \ket{0^{n}}$, uniformly sample $Q\in $ Orth$(2n) $ or Perm$(2n)$, {implement} the associated noisy Gaussian unitary $\widehat{U}_{Q}$ on $\rho_0$, and measure in {the} $Z$-basis with outcomes $x$;
\State Let $\hat{f}_{2k}^{(j)} := 2^n (-1)^k{n \choose k}^{-1} \supbraket{\bm 0|\Pcal_{2k}\Ucal_{Q}^{\dagger} | x}, \forall k\in [n]$;
\EndFor
\State $\hat f_{2k} := \MedianOfMeans\pbra{\cbra{f_{2k}^{(j)}}_{j = 1}^{R_c}, N_c, K_c}\forall k\in [n]$;
\State $\widehat\Mcal := \sum_{k=1}^n \hat{f}_{2k} \Pcal_{2k}$;
\State $R_e := N_eK_e$;
\For{$i\leftarrow 1$ to $m$}
\For{$j\leftarrow 1$ to $R_e$}
\State Prepare $\rho$, uniformly sample $Q\in O(2n) $ or Perm$(2n)$, 
 {implement} the associated noisy Gaussian unitary $\hat{U}_{Q}$ on $\rho$, and measure in the $Z$-basis with outcomes $x$;
\State Generate estimation $\hat{v}_i^{(j)} :=\tr{H_i \widehat\Mcal^{-1} \Ucal_{Q}^{\dagger}(x)}$;
\label{step:est_per}
\EndFor
\State $\hat{v}_i := \MedianOfMeans\pbra{\cbra{\hat{v}_i^{(j)}}_{j=1}^{R_e}, N_e, K_e}$;
\EndFor
\State \textbf{return}  $\cbra{\hat{v}_i}_{i = 1}^m$
\end{algorithmic}
\end{algorithm}

Incorporating the \textbf{MedianOfMeans} sub-procedure, as explained in Ref.~\cite{huang2020predicting}, guarantees that the number of quantum state copies needed relies on the logarithm of the number of observables. We included the \textbf{MedianOfMeans} sub-procedure in Supplementary Note 6 to ensure the completeness and consistency of this paper.
Our error analysis will involve selecting appropriate values for the number of calibrations and estimation samplings $N_c$, $N_e$, $K_c$, and $K_e$ to estimate the coefficients $\hat{f}_{2k}$ associated with the noisy channel.

Let $\hat{v}:= \tr{\widehat{\Mcal}^{-1} \Ucal^{\dagger}(\ketbra{x}{x})H}$ be an estimation of $\tr{\rho H}$ for some observables $H$ in the even subspace $\Gamma_{\mathrm{even}}$ and quantum states $\rho$, where $x$ follows the distribution $\tr{\ketbra{x}{x} \widehat{\Ucal}(\rho)}$ for $x\in \cbra{0,1}^n$,  then we have
\begin{align}
\abs{\hat{v} - \tr{\rho H}}
&\leq \abs{\hat{v} - \tr{\widehat{\Mcal}^{-1} \widetilde{\Mcal} \pbra{\rho}H}} \nonumber\\
&\quad + \abs{\tr{\widehat{\Mcal}^{-1} \widetilde{\Mcal}\pbra{\rho}H} - \tr{\rho H}}\nonumber\\
&= \varepsilon_e + \varepsilon_c ,
\label{eq:error_twoSource}
\end{align}
where $\varepsilon_e:= \abs{\hat{v} - \tr{\widehat{\Mcal}^{-1} \widetilde{\Mcal} (\rho)}}$ is the estimation error  and $\varepsilon_c:= \abs{\tr{\widehat{\Mcal}^{-1} \widetilde{\Mcal}\pbra{\rho}H} - \tr{\rho H}}$ is the calibration error. Therefore, by determining the necessary number of samples $N_e$ and $K_e$ to achieve the desired level of estimation error $\varepsilon_e$ (as well as $N_c$ and $K_c$ to account for the calibration error $\varepsilon_c$), we can obtain an estimation $\hat{v}$ with an overall error of $\varepsilon = \varepsilon_e + \varepsilon_c$, using $N_eK_e$ copies of the input state $\rho$.

\begin{theorem}
Let $\rho$ be an unknown quantum state and $\cbra{H_i}_{i=1}^m$ be a set of observables in the even subspace $\Gamma_{\mathrm{even}}$. Consider Algorithm \ref{alg:mitigated_estimation_alg} with the number of estimation samplings
\begin{align*}
\begin{aligned}
R_e &= \frac{68(1+\varepsilon_c)^2\ln (2m/\delta_e)}{2^{2n}\varepsilon_e^2} \sum_{0\leq l_1 + l_2 + l_3\leq n} g(l_1,l_2,l_3)\\
&\sum_{
\substack{S_1,S_2,S_3 \ \mathrm{ disjoint}\\
|S_i| = 2l_i \forall i\in [3]}
} \tr{\widetilde\gamma_{S_1} \widetilde\gamma_{S_2}H_0} \tr{\widetilde\gamma_{S_2} \widetilde\gamma_{S_3} H_0} \tr{\widetilde\gamma_{S_3}\widetilde\gamma_{S_1} \rho},
\end{aligned}
\end{align*}    
where $g(l_1,l_2,l_3) = \frac{(-1)^{l_1+l_2+l_3}{n\choose l_1, l_2, l_3}_p{2n\choose 2l_1 + 2l_2} {2n\choose 2l_2 + 2l_3}\Bcal_{l_1 + l_3}}{{2n\choose 2l_1, 2l_2, 2l_3}_p{n \choose l_1+l_2}{n \choose l_2 + l_3}\Bcal_{l_1 + l_2}\Bcal_{l_2+l_3}}$, and $H_0 = \max_i (H_i -\tr{H_i}\frac{\Ibb}{2^n})$, and
 the number of calibration samplings
 \begin{align*}
    R_c =\Ord{\frac{\Bcal_{\max}\sqrt{n}\ln n \ln (1/\delta_c)}{\Bcal_{\min}^2\varepsilon_c^2}},
 \end{align*}
where $\Bcal_{\max} = \max_k \abs{\Bcal_k}$ and $\Bcal_{\min} = \min_k \abs{\Bcal_k}$. 
Then, the outputs $\cbra{v_i}_{i=1}^m$ of the algorithm approximate $\cbra{\tr{\rho H_i}}_{i=1}^m$ with error $\varepsilon_e+\varepsilon_c$ and success probability $1-\delta_e-\delta_c$, under the assumption that $\vabs{H_i}_{\infty} = \Ord{1}$, where $\vabs{H_i}_{\infty}$ is the spectral norm of $H_i$.
\label{thm:est_cal_number}
\end{theorem}
We observe that the sampling for estimation we obtained is consistent with Wan et al.~\cite{wan2022matchgate} in the absence of noise. In the following, we will provide an analysis of the necessary number of measurements to compute $\braket{^{\bm k}\Dmat}$ using Algorithm \ref{alg:mitigated_estimation_alg}. To calculate the representation of each element $\braket{^{\bm k}\Dmat}_{j_1,\ldots, j_k; l_1,\ldots, l_k}$, where $j_i,l_i$ are in the range $[n]$ for $i\in [k]$, we need to calculate $m = \Ord{n^k}$ expectations for different $\widetilde\gamma_S$, where $\abs{S} = 2k$. By choosing the observable $H_j = \widetilde\gamma_{S}$ where $\abs{S}={2k}$ and $m = \Ord{n^k}$ in Theorem~\ref{thm:est_cal_number},
with the number of estimation samplings 
\begin{align}
R_e = \Ocal\pbra{\frac{kn^{k}\ln(n/\delta_e)}{\Bcal_{k}^2\varepsilon_e^2}}
\end{align}
and the number of calibration samplings 
\begin{align}
R_c = \Ord{\frac{\Bcal_{\max}\sqrt{n}\ln n 
\ln (1/\delta_c)}{\Bcal_{\min}^2\varepsilon_c^2}},
\label{eq:rc_supp}
\end{align}
 the estimation error can be bounded to $\varepsilon_e+ \varepsilon_c$. 
The equations for $R_e$ and $R_c$ can be simplified to $R_e = \Ocal\pbra{\frac{kn^k\ln(n/\delta_e)}{\varepsilon_e^2}}$ and $R_c = \Ord{\frac{\sqrt{n} \ln n \ln (1/\delta_c)}{\varepsilon_c^2}}$ for the general noises with constant average fidelity $\Bcal_{k}$ in subspace $\Gamma_{2k}$ for any $k\in\cbra{0,\ldots, n}$. We give more details for the calculations in Supplementary Note 8.
However, some types of noise channels, such as certain Gaussian unitary channels present in the related $2n\times 2n$ matrix $Q$, cannot be mitigated with our mitigation algorithm. In particular, there exists a signed permutation matrix $Q$ for which $f_{2k} = 0$, resulting in complete loss of projection for the observable onto the subspace $\Gamma_{2k}$.
As a result, it is impossible to calculate $\tr{\rho \widetilde{\gamma}_S}$ for any set $S$ containing $2k$ elements. We anticipate that the noise in the quantum device will differ significantly from the $U_Q$ which belongs to the intersection of the matchgate and Clifford groups.

\subsection*{Numerical results}
\label{sec:numerical_res}
Here we give the numerical results for the mitigated shadow estimation in the fermionic systems. 
Since the elements of a $k$-RDM can be expressed in the form $\tr{\rho {\gamma}_S}$, we give the numerical results of the errors of the estimators for the expectation value of local fermionic observables $\tr{\rho \widetilde{\gamma}_S}$.


\begin{figure*}[t]
    \centering
    \includegraphics[width =1.0\textwidth]{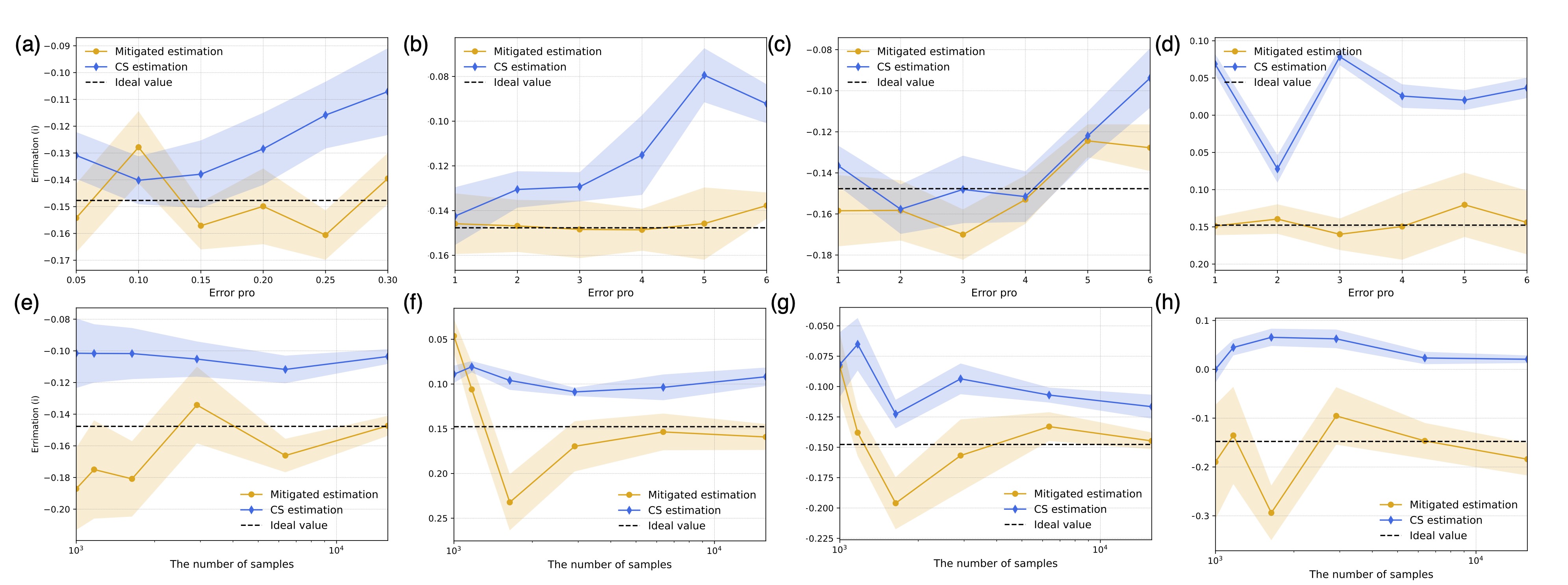}
    \caption{{The estimations of the expectation values of Majorana operators. (a-d) The estimations with the changes of the noise strength and (e-h) the number of samples for fixed noise parameters, where (a) and (e) are associated with depolarizing noise,  (b) and (f) are associated with generalized amplitude damping noise, (c) and (g) are associated with X-rotation noise, (d) and (h) are associated with Gaussian unitary noise.
The error bar is the estimation of the standard deviation by repeating the procedure for $R = 10$ rounds, and it is estimated to be $ \sqrt{\sum_{i = 1}^R (\hat{v}_i-\bar{v})^2/R}$, where $\bar{v}$ is the expectation of the corresponding estimations. Specially, $\bar{v}=\frac{f_{|S|}}{g}\tr{\rho \gamma_S}$, $g= \hat{f}_{\abs{S}}$ for mitigated estimation, and $g=\frac{\binom{n}{k/2}}{\binom{2n}{k}}$ for CS estimation. In Fig.~\ref{fig:majorana} (a-d), the $x$-axis represents the labels for different noise parameters. The dotted horizontal line indicates the value $\tr{\rho \widetilde{\gamma}_{S}}$.}}
    \label{fig:majorana}
\end{figure*}

The estimator $\hat{v}$ for $\tr{\rho \widetilde{\gamma}_S}$ can be represented as in Eq.~\eqref{eq:est_majoranaOpe}.
Here we choose the number of qubits $n = 4$ and $S = \cbra{1,2}$, and $\widetilde{\gamma}_S = U_Q^\dagger \gamma_S U_Q$ and $Q$ is uniformly randomly chosen from Perm$(2n)$. The quantum state $\rho$ is a uniformly randomly generated $4$-qubit pure state. As shown in Fig.~\ref{fig:majorana}, we depict the {estimations} of classical shadow estimators~\cite{wan2022matchgate} and our mitigated Algorithm \ref{alg:mitigated_estimation_alg}, with the changes of noise strength ({Fig.~\ref{fig:majorana}(a-d)}) and the changes of the number of samples ({Fig.~\ref{fig:majorana}(e-h)}). Here we numerically test the estimations with respect to depolarizing, amplitude damping, $X$-rotation, and Gaussian unitaries.

The calibration samples {in \Cref{alg:mitigated_estimation_alg} for the numerics} are $N_c = 4,000$ and $K_c = 20$ for all noise models.
The number of samples for the classical shadow method is set as $N_e = 4,000, K_e =10$, and for \Cref{alg:mitigated_estimation_alg} are set as $N_e = 4,000/(1-p_{\text{noise}})$ and $K_e = 10$, where $p_{\text{noise}}$ is the noise strength varying for different noise settings:

(1) Depolarizing noise $\Lambda_{\mathrm d}(\rho) = (1-p)\rho + p\frac{\Ibb}{2^n}$, where $\rho$ is any quantum state and $p$ is the depolarizing noise strength. In Fig.~\ref{fig:majorana} (a), $p$ varies from 0.05 to 0.3 ($p =0.05j$ for $x$-axis equals $j$ where $j\in [6]$), and in Fig.~\ref{fig:majorana} (b), $p = 0.2$. 

(2) {Generalized} amplitude damping noise $\Lambda_{\mathrm a}$ with
representation
\begin{align}
\Lambda_{\mathrm a}(\rho) = \sum_{
\substack{u, v\in \cbra{0,1}^n\\
u\ne v}
} E_{uv} \rho E_{uv}^{\dagger} + E_0 \rho E_0^{\dagger},
\label{eq:gen_amp_noise}
\end{align}
where $E_{uv} = \sqrt{ p_{uv}}\ketbra{v}{u}$ for $u\neq v \in \{0,1\}^n$ and $E_0 = \sqrt{\Ibb - \sum_{
\substack{u,v\in \cbra{0,1}^n\\
u\ne v}
}E_{uv}^{\dagger}E_{uv}}$.
Note that Eq.~\eqref{eq:gen_amp_noise} is a generalization of Eq.~\eqref{eq:amp_noise}, which connects to  Eq.~\eqref{eq:amp_noise} by setting $p_{uv} = \bar{p}_u$ for any $u,v\in [n]$.
Here we uniformly randomly choose $p_{uv}$ in {$\frac{[0,1]+ j-1}{6\times 2^{n+1}}$}, and labeled it as $j$ in the $x$-axis of Fig.~\ref{fig:majorana} where $j\in [6]$, and choose the generated damping errors for {$j=5$}  case in {Fig.~\ref{fig:majorana} (b) as the damping errors for Fig.~\ref{fig:majorana} (f)}. 

(3) $X$-rotation noise $\Lambda_{\mathrm r}$ defined in Eq.~\eqref{eq:Rx_noise} with noise parameters {$\theta_j = \frac{\pi}{2\pbra{8-j}}$} for $j\in [6]$, and the noise strength is chosen as $1-\cos \theta$ in {Fig.~\ref{fig:majorana} (c), and Fig.~\ref{fig:majorana} (g)} depicts the errors in the $X$-rotation noise with noise parameter $\theta_6$. The $x$-axis of the mitigation results of $X$-rotation noise in {Fig.~\ref{fig:majorana} (c)} denotes the label for noise parameters $\cos \theta_j$ for $j$ range from 1 to 6.

{(4) The Gaussian unitary noise channel $\mathcal{U}_Q$ is chosen such that $Q$ is sampled from the signed permutation group, ensuring that the coefficient $f_{1}$ for the noise channel is non-zero.
The associated numerical results are shown in {Fig.~\ref{fig:majorana} (d) and (h)}.
The number of estimation samplings $N_e = 8000, K_e = 10$ for Fig.~\ref{fig:majorana} (d)}.
{Fig.~\ref{fig:majorana} (h)} choose the same noise parameter with the fifth noise parameter in {Fig.~\ref{fig:majorana} (d)}.
From the figure, we see that without mitigation, the error is enormous with the CS algorithm.

{In Fig.~\ref{fig:majorana} (e-h) the number of samples ranges from $\floor{900+100\exp(j)}$ for $j\in \cbra{0,\ldots, 5}$.} From Fig.~\ref{fig:majorana}, we see that with the increase of the noise strength, the classical shadow method with depolarizing, amplitude damping, and $X$-rotation noise all gradually diverge to the expected value $\tr{\rho \widetilde\gamma_S}$, while our error-mitigated estimation protocol in Alg.~\ref{alg:mitigated_estimation_alg} gives an expected value that is close to the noiseless value. {Based on the numerical results depicted in Fig.\ref{fig:majorana} (e-h), it is evident that as the number of samples increases, the estimation outcomes generated by Algorithm~\ref{alg:mitigated_estimation_alg} approaches the expected value $\tr{\rho \widetilde{\gamma}_S}$}, while the convergent value for the classical shadow method is far from the expected value with depolarizing, amplitude damping, and $X$-rotation noises. {Conversely, the error bar associated with CS estimations is observed to be smaller compared to mitigated estimations, using the same number of samplings, as illustrated in Fig. \ref{fig:majorana} (e-h). This is due to the variance of the mitigated estimations associated with the average noise fidelities in $\Gamma_{2k}$ subspace $\Bcal_k$.}


\section*{Discussion}
\label{sec:dis}

We present an error-mitigated classical shadow algorithm for noisy fermionic systems, thereby extending matchgate classical shadows for noiseless systems~\cite{Zhao21Fermionic,wan2022matchgate}. With our method, the calibration process requires a number of copies of the classical state $\ketbra{\bm 0}{\bm 0}$ that scales logarithmically with the number of qubits. Assuming a constant average noise fidelity for the noise channel, our algorithm requires the same order of estimation copies as the matchgate classical shadow without error mitigation~\cite{wan2022matchgate}. Our algorithm is applicable for efficiently calculating all the elements of a given $k$-RDM.

To provide a clearer demonstration of the average fidelity of common noises, we consider depolarizing, amplitude damping, and $X$-rotation noises. The average fidelity of depolarizing and amplitude damping noises are given by $(1-p)$ and $(1-\sum_{u}\bar{p}_u)$ respectively, where $p,\bar{p}_u$ are the noise parameters. For $X$-rotation noise, the average fidelity lies between $[\min_{\theta}\cos^k\theta ,\max_{\theta}\cos^k\theta]$ where $\theta$ is the rotated angles. To evaluate the effectiveness of our algorithm in mitigating these noises, we compare its performance with the matchgate CS algorithm to calculate the expectations of $\widetilde\gamma_S$, where $\abs{S} = 2k$, which is crucial for calculating $k$-RDMs. Our numerical results show good agreement with the theory, validating the effectiveness of our algorithm.

While our algorithm demonstrates good performance in the presence of common types of noise in near-term quantum devices, further investigations are required to explore its potential limitations and improvements. Some of the open questions that can be addressed in future research include:

\begin{itemize}
    \item [(1)] Is it possible to extend our algorithm to handle other types of noise, such as time-dependent, non-Markovian and environmental noise~\cite{van2006introduction}, or more generally, noise that does not satisfy the GTM assumption? If so, how would these different types of noise impact the performance of our algorithm?
    {
\item[(2)] We provide numerical results under the assumption of Gaussian unitary noise, a common noise model in the fermionic platform. An intriguing unanswered question pertains to the performance of our algorithm in the presence of more typical noise channels inherent to fermionic platform.
    }
    \item[(3)] The number of gates required by the matchgate circuit is $\Ocal(n^2)$~\cite{Jiang18Quantum}. As a result, the accumulation of noise significantly increases the error mitigation threshold~\cite{Deshpande22tight,quek2022exponentially}, which raises the intriguing question of whether it is feasible to provide an error-mitigated classical shadow using a shallower circuit. This may be compared with Bertoni et al.~\cite{bertoni2022shallow}, who propose a shallower classical shadow approach for qubit systems.
    \item[(4)] In addition, we have included in the Supplementary materials the analysis and numerical results regarding the overlap between a Gaussian state and any quantum states, as well as the inner product between a Slater determinant and any pure state. This prompts the question of whether our algorithm can be utilized to calculate other physical, chemical, or material properties beyond the scope of this paper.
\end{itemize}

Exploring these questions would enhance our comprehension of the potential and limitations of our algorithm, and could potentially pave the way for advancements in the estimation of fermionic Hamiltonian expectation values with near-term quantum devices.

\emph{Note added.---} Following the completion of our manuscript, we became aware of recent independent
work by Zhao and Miyake~\cite{zhao2023grouptheoretic}, 
who also study ways to counteract noise in the
fermionic shadows protocol.

\section*{Methods}
\label{sec:learn_noisy_channel}

\subsection*{Noisy fermionic channel representation}
Here we present an unbiased estimation approach for the noisy representation of the fermionic shadow channel, which utilizes a protocol similar to the matchgate benchmarking protocol~\cite{Helsen22Matchgate}. According to representation theory~\cite{fulton2013representation} (see details in Supplementary Note 1), the noisy fermionic shadow channel can be represented as $\widetilde\Mcal = \sum_{k = 0}^n f_{2k}\Pcal_{2k}$. 
Since $\tr{H \widetilde{\Mcal}(\rho)} = \tr{\rho \widetilde{\Mcal}(H)}$, with the pre-knowledge of $f_{2k}$ we can calculate $\tr{H \widetilde{\Mcal}(\rho)}$ for any observable $H$ in the even subspace. To learn the $2(n+1)$ coefficients, we begin with the easily prepared state $\rho_0 = \ketbra{\bm 0}{\bm 0}$ and apply a noisy unitary channel $\widehat\Ucal$ with $\Ucal$ sampled from the matchgate group. We then perform a $Z$-basis measurement $\Xcal$ with measurement outcomes $x\in\cbra{0,1}^n$, followed by classically operating the unitary channel $\Ucal^{\dagger}$ on $\ketbra{x}{x}$. The generated state has expected value $\sum_{k = 0}^{n} f_{2k}\Pcal_{2k}(\rho_0)$,
and $f_{2k}$ is obtained by projecting the final state to the $\Pcal_{2k}$ subspace with some classical post-processing. We illustrate the learning process of the noisy channel in Fig.~\ref{fig:scheme} (a). The following theorem provides an unbiased estimation of the noisy fermionic classical shadow.

\begin{theorem}
    The noisy fermionic shadow channel can be represented as $\widetilde{\Mcal} = \sum_{k} f_{2k} \Pcal_{2k}$, where $\Pcal_{2k}$ is defined in Eq.~\eqref{eq:projection_fermionic_system}, and $\hat{f}_{2k} = {2^{n}}\supbraket{\bm 0|\Pcal_{2k} \Ucal_{Q}^\dagger | x}/{{n\choose k}}$ is an unbiased estimator of $f_{2k}\in \Rbb$,
where $\supket{x}$ is the measurement outcome from the noisy shadow protocol obtained by starting from the input state $\supket{\bm 0}$ and applying a noisy quantum circuit $\widetilde\Ucal_{Q}$ followed by a $Z$-basis measurement, where  $\Ucal_{Q}$ is uniformly randomly picked from the matchgate group.
\label{thm:cali_f_est}
\end{theorem}

The representation of the noisy fermionic channel, denoted by $\widetilde{\Mcal}=\sum_k f_{2k}\Pcal_{2k}$, where $f_{2k}\in \Cbb$, can be obtained by the irreducible representation of the Gaussian unitary. A detailed proof of this theorem is provided in Supplementary Note 4. We claim that the coefficients $\hat{f}_{2k}$ can be efficiently calculated with the following lemma.
\begin{lemma}
$\hat{f}_{2k}$ is the coefficient of $x^{k}$ in the polynomial $p_{Q}(x)$, where
\begin{equation}
p_{Q}(x)= {n \choose k}^{-1}\pf\pbra{C_{\ket{\bm 0}}}\pf\pbra{-C_{\ket{\bm 0}}^{-1} + x Q^TC_{\ket{x}}Q},
\label{eq:est_f_polycal}
\end{equation}
where $C_{\ket{x}} = \bigoplus_{i = 1}^n \begin{pmatrix}
0 & (-1)^{x_j}\\
(-1)^{x_j + 1} & 0
\end{pmatrix}$
 is the covariance matrix of $\ket{x}$.
\label{lem:f_2k_calculation}
\end{lemma}
This lemma can be obtained by Proposition 1 in Ref.~\cite{wan2022matchgate}. For the completeness of this paper, we also give the proof of this lemma in Supplementary Note 4. The coefficients can be calculated with the polynomial interpolation method in polynomial time. 
With Theorem~\ref{thm:cali_f_est}, we can give an unbiased estimation $\widehat{\Mcal}$ for $\widetilde{\Mcal}$. 

By the definition of $\hat f_{2k}$, and the twirling properties of $\int_Q d \mu(Q) {\Ucal_{Q}^{\otimes 2}}$, the expectation value  for the estimation $\hat f_{2k}$ can be formulated as
\begin{align}
f_{2k} = {2n\choose 2k}^{-1}{n\choose k}\Bcal_k.
\end{align}
We postpone the details of this proof to Supplementary Note 4. It implies that in the noiseless scenario, $\widetilde{M}$ degenerates into $\Mcal$ as defined in Equation \eqref{eq:fermion_channel}.
Combined with the definition of $\Bcal_{k}$ in Eq.~\eqref{eq:def_B_k}, $f_{2k}$ is close to ${2n\choose 2k}^{-1}{n\choose k}$ if the average noise fidelity in $\Gamma_{2k}$ subspace is close to one. 
Recall that $\Bcal_{k}$ is a constant in the depolarizing, amplitude-damping, and $X$-rotation noises with a constant noise strength, which implies that these noises
can be efficiently mitigated with our algorithm. 

Alternatively, we have a counterexample in Supplementary Note 3 that illustrates the limitations of our mitigation algorithm. Specifically, if the noise follows a Gaussian unitary channel $\Ucal_Q$ where $Q$ is a signed permutation matrix (associated with a discrete Gaussian unitary), then $\Bcal_k$ may become zero. Hence, $f_{2k} = 0$, rendering our estimation approach unsuitable.

Recall that our goal is to estimate $\cbra{\tr{\rho H_i}}_{i=1}^m$ using a noisy quantum device and polynomial classical cost, where $\rho$ is an $n$-qubit quantum state and $H_i$ is an observable in the even subspace $\Gamma_{\mathrm{even}}$.
Here we visualize the estimation process with the guarantee in Theorem~\ref{thm:cali_f_est}.
We uniformly randomly sample a matchgate $U_{Q}$ from the matchgate group and apply it to the quantum state $\rho$, and then measure in the $Z$-basis to get outcomes $x$. We define the estimator
\begin{align}
\hat{v} &= \tr{H{\widehat\Mcal^{-1}(U_Q^\dagger \ketbra{x}{x}U_Q)}}\\
&=\sum_{k=0}^n \hat{f}_{2k}^{-1}\tr{H \Pcal_{2k}\pbra{U_Q^\dagger \ketbra{x}{x}U_Q}}.
\label{eq:estimator}
\end{align}
{
It is easy to show that $\text{tr}(H \widehat{\mathcal{M}}^{-1}\widetilde{\mathcal{M}}(\rho))$ is an unbiased estimation of $\text{tr}(\rho H)$ when $\widetilde{\mathcal{M}}$ is invertible, specifically when $f_{2k}\neq 0$ for any $k$, and $H$ belongs to the even subspace $\Gamma_{\mathrm{even}}$.
Given that $\mathbb{E}[\widehat{\mathcal{M}}] = \widetilde{\mathcal{M}}$ and $\hat{v}$ serves as an unbiased estimator of $\text{tr}(H \widehat{\mathcal{M}}^{-1}\widetilde{\mathcal{M}}(\rho))$, it implies that the estimation error $\varepsilon := |\hat{v} - \text{tr}(\rho H)|$ is bounded by $|\hat{v} - \text{tr}(H \widehat{\mathcal{M}}^{-1}\widetilde{\mathcal{M}}(\rho))| + |\text{tr}(H \widehat{\mathcal{M}}^{-1}\widetilde{\mathcal{M}}(\rho)) - \text{tr}(\rho H)|$, which can be minimized with the increasing of the number of samplings for the estimations $\hat{v}$ and $\text{tr}(H \widehat{\mathcal{M}}^{-1}\widetilde{\mathcal{M}}(\rho))$, as shown in Eq. \eqref{eq:error_twoSource}.} Note that the estimator defined in Eq.~\eqref{eq:estimator} is not always efficient for all states $\rho$ and observables $H$. 
Here we claim that with this estimator, we can efficiently calculate substantial physical quantities such as the expectation value of $k$-RDM, which not only serves the variational quantum algorithm (VQE) of a fermionic system with up to $k$ particle interactions~\cite{malone2022towards,liu2021efficient}, but also provide supports to the calculations of derivatives of the energy~\cite{overy2014unbiased,o2019calculating} and multipole moments~\cite{rubin2018application}. It is also an indispensable resource for the error mitigation technique~\cite{McClean17Hybrid,Takeshita20Increasing}. It also serves to calculate the overlap between a Gaussian state and any quantum state, and the inner product between a Slater determinant and any pure state inspired by the fermionic shadow analysis of Wan et al.~\cite{wan2022matchgate}. We postpone the details to Supplementary Note 5.

Note that all elements of $k$-RDMs can be derived through $\tr{\rho \gamma_S}$ for a total of $\Ord{n^k}$ sets $S$ with $|S| = 2k$.  In an expansion of this concept, we now focus on evaluating $\tr{\rho \widetilde \gamma_S}$ for $\Ord{n^k}$ different $S$ with $|S| = 2k$.
To calculate the expectation value $\tr{\rho \widetilde\gamma_S}$, we set the input quantum state to be $\rho$ and the observable to be $H = \widetilde{\gamma}_S$ in the estimation formula of Eq.~\eqref{eq:estimator}, which can then be simplified to
\begin{align}
\hat{v} = i^{k} \hat{f}_{2k}^{-1}\pf\pbra{Q_1Q^T C_{\ket{x}}QQ_1^T|_S},
\label{eq:est_majoranaOpe}
\end{align}
where $\widetilde{\gamma}_j = \sum_{l\in [2n]} Q_{1}(j,l)\gamma_l$, $Q_{1}(j,l)$ is the $(j,l)$-th element of $Q_1$, and $C_{\ket{x}} = \bigoplus_{i = 1}^n \begin{pmatrix}
0 & (-1)^{x_j}\\
(-1)^{x_j + 1} & 0
\end{pmatrix}$ is the covariance matrix of $\ketbra{x}{x}$. Here, $A|_S$ refers to the submatrix obtained by taking the columns and rows of the matrix $A$ that are indexed by $S$. The simplified quantity can be calculated in polynomial time since $\hat{f}_c$ and the Pfaffian function can be calculated efficiently. We give a detailed proof of the simplification process in Supplementary Note 5.

\section*{Data Availability}
The datasets produced during the current study are available at \href{https://github.com/GillianOoO/Error-mitigated-fermionic-classical-shadow.git}{https://github.com/GillianOoO/Error-mitigated-fermionic-classical-shadow.git}.

\section*{Acknowledgments}
We express our appreciation for the valuable input and feedback given by Jens Eisert, Janek Denzler, and Ellen Derbyshire. We are also thankful for the enlightening conversations held with Xiao Yuan and Yukun Zhang. Furthermore, we extend our gratitude to Janek Denzler and Ellen Derbyshire for their assistance with certain numerical aspects. We also thank Guang Hao Low for their comments on an earlier version of this manuscript. 
BW acknowledges funding support from the Bundesministerium für Bildung und Forschung (FermiQP, DAQC), Bundesministerium für Wirtschaft und Klimaschutz (EniQmA), and the Munich Quantum Valley (K-8).
DEK acknowledges funding support from the Agency for Science, Technology and Research (A*STAR) Central Research Fund (CRF) Award, A*STAR C230917003, and the National Research Foundation, Singapore and A*STAR under its Quantum Engineering Programme (NRF2021-QEP2-02-P03).

\section*{Competing Interests}
All authors declare no financial or non-financial competing interests.

\section*{Author contributions}
Bujiao Wu and Dax Enshan Koh conducted, discussed, and wrote the paper together. Bujiao Wu performed the numerical simulation.


\newpage

\appendix

\makeatletter
\newcommand*{\rom}[1]{\expandafter\@slowromancap\romannumeral #1@}
\makeatother

\begin{center}
{\bf\large Supplementary Notes}
\end{center}

\section*{Supplementary Note 1}
\label{app:group_properties}

This section provides an overview of group theory~\cite{goodman2009symmetry,fulton2013representation} and its application in the context of the matchgate group. 

Let $\Gmat$ be a finite group and let $U(V)$ denote the group of unitary linear transformations on a finite complex vector space $V$. We define a representation $\phi$ of the group $\Gmat$ on the space $V$ as
\begin{equation}
\phi: \Gmat\rightarrow U(V):\Gmat\rightarrow \phi(\Gmat).
\end{equation}
Maschke's lemma shows that any representation of a group can be uniquely represented as a direct sum of its irreducible representations, i.e.
\begin{equation}
\phi(G) = \bigoplus_{\lambda\in R_G} \phi_{\lambda}(G)^{\otimes m_{\lambda}},
\end{equation}
for any $G\in \Gmat$, where $\phi_\lambda$ is the irreducible representation of $\Gmat$ and $m_{\lambda}$ is the multiplicity of $\phi_{\lambda}$. Here we consider only the multiplicity-free case where $m_{\lambda} = 1$ for all $\phi_{\lambda}$. We provide a corollary of Schur's lemma to evaluate twirls over multiplicity-free representations as follows.

\begin{lemma}[Rephrase of {\cite[Lemma 1.7 and Proposition 1.8]{fulton2013representation}} and {\cite[Lemma 1]{helsen2019new}}]
Let $\Gmat$ be a finite group, and $\phi$ be a multiplicity-free representation of $\Gmat$ on a finite complex vector space $V$ with decomposition
\begin{equation}
\phi(G) \simeq \bigoplus_{\lambda \in R_G} \phi_{\lambda}{G}
\end{equation}
for any $G\in \Gmat$, into inequality irreducible representation $\phi_{\lambda}$. Then for any linear map $Q: V \rightarrow V$, the twirl of $Q$ over $\phi$ has the form
\begin{equation}
\Ecal_{\phi}(Q) = \frac{1}{\abs{\Gmat}} \sum_{G \in \Gmat} \phi(G) Q \phi(G)^\dagger = \sum_{\lambda\in R_G} \frac{\tr{Q\Pcal_{\lambda}}}{\tr{\Pcal_{\lambda}}}\Pcal_{\lambda} ,
\end{equation}
where $\Pcal_{\lambda}$ is the projector onto the support of the representation $\phi_{\lambda}$.
\label{lem:schur_lemma}
\end{lemma}

Let $Q^{(u,v)}\in \Rbb^{2n\times 2n}$ be the basic rotation matrix with $Q^{(u,v)} = Q^{(u,v)}_{uv} \oplus \Ibb_{-uv}$, where the submatrix $Q^{(u,v)}_{uv}$ is constructed by the $u$th and $v$th rows and columns of $Q^{(u,v)}$, and is equal to
\begin{equation}
Q^{(u,v)}_{uv} = \begin{pmatrix}
\cos \theta & -\sin \theta\\
\sin \theta & \cos \theta
\end{pmatrix}_{u,v}.
\end{equation}
The submatrix $\Ibb_{-uv} = \Ibb - \ketbra{u}{u} - \ketbra{v}{v}$ contains all of the rows and columns of identity except the $u$-th and $v$-th rows, and columns. We define matchgate $U_{Q^{(u,u+1)}}$corresponding to the rotation matrix $Q^{(u,u+1)}$ of two nearest neighbors as
$\exp(i\theta X_{u/2}X_{u/2+1})$ for even $u$, and $\exp(i\theta Z_{(u+1)/2})$ for odd $u$. Let the reflection with angle $0$ to the $2k$th row and column, denote as $Q^{(k)} = [-1]_{2k} \oplus \Ibb_{-2k}$, and the matchgate $U_{Q^{(k)}}=X_{k}$. Hence 
all of $\cbra{Q^{(u,u+1)}, Q^{(k)}}$ generates the orthogonal group $\text{Orth}(2n)$, and
\begin{equation}
U_Q^\dagger \gamma_u U_Q = \sum_{v = 1}^n Q_{uv}\gamma_v, 
\end{equation}
for any $u\in [2n]$.
The matchgate group is denoted as
\begin{equation}
\Mbb_n = \cbra{U_Q|Q\in \text{Orth}(2n)}.  
\end{equation}

As an application of Lemma
\ref{lem:schur_lemma}, we define the superoperator representation of a given unitary as $\phi(U):= \Ucal$. Then,
\begin{equation}
\phi(U)\phi(V) = \phi(UV).
\end{equation}

In the following, we review the results of the first three moments of the uniform distribution over the matchgate group.
\begin{lemma}[Wan et al. \cite{wan2022matchgate}]
Let $Q$ be a matrix that is uniformly randomly sampled from the orthogonal group $\text{Orth}(2n)$. Then,
\begin{align}
& \int_Q d \mu(Q){\Ucal_Q} = \supket{\Ibb}\supbra{\Ibb}\\
& \int_Q d \mu(Q){\Ucal_Q^{\otimes 2}} = \sum_{k=0}^{2n}\supket{\Rcal_k^{(2)}}\supbra{\Rcal_k^{(2)}}
\\
& \int_Q d \mu(Q){\Ucal_Q^{\otimes 3}} = \sum_{
\substack{
k_1,k_2,k_3\geq 0\\
k_1 + k_2 + k_3\leq 2n
}
}\supket{\Rcal_{k_1,k_2,k_3}^{(3)}}\supbra{\Rcal_{k_1,k_2,k_3}^{(3)}},
\end{align}
where $\mu$ is the normalised Haar measure on $\mathrm{Orth}(2n)$, and 
\begin{align}
    &\supket{\Rcal_k^{(2)}} = { 2n\choose k}^{-1/2} \sum_{S\subseteq [2n],|S|=k}\supket{\gamma_S}\supket{\gamma_S}
    \label{eq:sec_moment}
    \\ &\supket{\Rcal_{k_1,k_2,k_3}^{(3)}} = {2n \choose k_1,k_2,k_3,2n-k_1-k_2-k_3}^{-1/2} \sum_{
    \substack{
 S_1, S_2, S_3\subseteq [2n] \ \mathrm{disjoint}\\
 \abs{S_j}=k_j \forall j \in [3]
    }
    } \supket{\gamma_{S_1}\gamma_{S_2}}
    \supket{\gamma_{S_2}\gamma_{S_3}}
    \supket{\gamma_{S_3}\gamma_{S_1}}.
\label{eq:third_moment}
\end{align}
\label{lem:threemomentsMatchgate}
\end{lemma}
It is easy to see that we can further generalize $\supket{\Rcal_k^{(2)}},\supket{\Rcal_{k_1,k_2,k_3}^{(3)}}$ in Eqs.~\eqref{eq:sec_moment} and \eqref{eq:third_moment} to $\supket{\widetilde\Rcal_k^{(2)}},\supket{\widetilde\Rcal_{k_1,k_2,k_3}^{(3)}}$ in the basis of $\widetilde{\gamma}_S$. Specifically,
\begin{align}
    &\supket{\widetilde\Rcal_k^{(2)}} = { 2n\choose k}^{-1/2} \sum_{S\subseteq [2n],|S|=k}\supket{\widetilde\gamma_S}\supket{\widetilde\gamma_S}\\
    &\supket{\widetilde\Rcal_{k_1,k_2,k_3}^{(3)}} = {2n \choose k_1,k_2,k_3,2n-k_1-k_2-k_3}^{-1/2} \sum_{
    \substack{
 S_1, S_2, S_3\subseteq [2n] \ \mathrm{disjoint}\\
 \abs{S_j}=k_j \forall j \in [3]
    }
    } \supket{\widetilde\gamma_{S_1}\widetilde\gamma_{S_2}}
    \supket{\widetilde\gamma_{S_2}\widetilde\gamma_{S_3}}
    \supket{\widetilde\gamma_{S_3}\widetilde\gamma_{S_1}}.
\label{eq:tilde_third_moment}
\end{align}

\section*{Supplementary Note 2}
\label{app:pro_gaussian}
Two states that characterize a fermionic system are the Gaussian state and the Slater determinant. Gaussian states are ground or thermal states of fermionic systems, and all information about the state is encoded in its $2N \times 2N$ Hamiltonian. The Slater determinant describes the wave function of a many-particle (fermionic) system. We give formal definitions of them in the following. 

\textbf{Gaussian state.} The Gaussian state is defined as
\begin{equation}
\rho_g = \prod_{k = 1}^n \frac{(\Ibb - i \mu_k \widetilde{\gamma}_{2k-1} \widetilde{\gamma}_{2k})}{2}
\label{eq:gaussian_state_rep}
\end{equation}
for some coefficients $\mu_k\in [-1,1]$, where $\widetilde{\gamma}_k = U_Q^{\dagger} \gamma_k U_Q$ for some $Q\in \mathrm{Orth}(2n)$. The covariance matrix for a Gaussian state is a $2n\times 2n$ matrix with the $(k,l)$-th element being $-i\tr{[\gamma_k,\gamma_l]\rho_g}$. The computational basis $\ketbra{x}{x}$ is a special Gaussian state with the representation
\begin{equation}
\ketbra{x}{x} = \prod_{k=1}^n \frac{1}{2}(\Ibb - i (-1)^{x_j} \gamma_{2j-1}\gamma_{2j}),
\end{equation}
and the covariance matrix representation $C_{\ket{x}}$ is equal to
\begin{equation}
C_{\ket{x}} = \bigoplus_{i = 1}^n \begin{pmatrix}
0 & (-1)^{x_j}\\
(-1)^{x_j + 1} & 0
\end{pmatrix}.
\end{equation}
The covariance matrix for a general Gaussian state defined in Eq.~\eqref{eq:gaussian_state_rep} is
\begin{align}
C_{\rho_g} = Q \bigoplus_{j = 1}^n \begin{pmatrix}
0 & \mu_j\\
-\mu_j & 0
\end{pmatrix} Q^T.
\end{align}

\textbf{Slater determinant.} The $\tau$-fermion Slater determinant $\ket{\phi_\tau}$ is defined as
$\ket{\phi_\tau}= \tilde{b}_1^\dagger \cdots \tilde{b}_\tau^\dagger \ket{\bm 0}$, and $\tilde{b}_k = \sum_{l=1}^n U_{kl}b_l$ for some unitary $U\in \Cbb^{n\times n}$.

\section*{Supplementary note 3}
\label{app:B_k_bound}
Recall that the fermionic shadow scheme relies on a twirl of the full space that divides into $2n+1$ even subspaces $\Gamma_{2k}$ for $k \in \cbra{0, ..., n}$.
In this section, we will investigate the range of the average fidelity parameters $\Bcal_k$ under different noise scenarios, where $k\geq 1$. To start, we will prove that in (0) noise-free case, the average noise fidelities $\Bcal_k$ equals one for any $k$.

The noise channels under investigation, which are also discussed in the main text, represent typical noise models applicable to various architectures, containing (1) depolarizing noise, a channel that gives the probability of returning the state itself or the maximally mixed state; (2) amplitude-damping noise, a model of physical processes such as spontaneous emission; (3) $X$-rotation noise, which describes the rotation offset around the $x$-axis; and (4) Gaussian unitary noise, which is common in continuous variable systems and has the form of a normal distribution~\cite{WZ+21, Weedbrook+2011}.

Since $\Lambda$ is a 
completely positive and trace-preserving
map, we can write
\begin{align}
\Lambda(\cdot) = \sum_{l} E_l (\cdot ) E_l^{\dagger}
\end{align}
in its Kraus decomposition such that $\sum_{l} E_l^{\dagger}E_l = \Ibb$. In the following, we analyze the range of $\Bcal_k$ for the noise-free model and different common noise models.
\begin{itemize}
    \item [(0)] By the definition in Eq. (5) in the main file, for a noise-free model where $\Lambda$ implements the identity transformation $\Ical$, 
\begin{align}
\Bcal_k &= 2^{-n} (-i)^k {n \choose k}^{-1} \sum_{x} \sum_{S\in {[n]\choose k}} (-1)^{\sum_{j\in S} x_j} (-i)^{k} (-1)^{\sum_{j\in S}x_j} \tr{\gamma_{D(S)}^2}2^{-n}\\
&= 1.
\end{align}
\item[(1)] For the depolarizing noise with channel representation $\Lambda_{\mathrm d}(A) = (1-p)A + p\tr{A}\frac{\Ibb}{2^n}$ for any linear operator $A$, where $p\in [0,1]$ is the error probability, we have
\begin{align}
\Bcal_k &= \frac{(-i)^k}{2^n{n\choose k}} \sum_x \sum_{S\in {[n]\choose k}}(-1)^{\sum_{j\in S}x_j}(1-p)(-i)^k(-1)^{\sum_{j\in S} x_j}(-1)^k\\
&= 1-p.
\end{align}
\item[(2)] The {general} amplitude-damping noise can be denoted as the following equation with Kraus decomposition representation
\begin{align}
\Lambda_{\mathrm a}(A) = \sum_{
\substack{u, v\in \cbra{0,1}^n\\
u\ne v}
} E_{uv} A E_{uv}^{\dagger} + E_0 A E_0^{\dagger},
\end{align}
where $E_{uv} = \sqrt{p_{uv}}\ketbra{v}{u}$, and {$E_0 = \sqrt{\Ibb - \sum_{
\substack{u,v\in \cbra{0,1}^n\\
u\ne v}
}E_{uv}^{\dagger}E_{uv}}$}, and the probabilities satisfy $\sum_{v\in \cbra{0,1}^n\backslash \cbra{u}}p_{uv}\leq 1$. 
Here we only analyze the bound of $\Bcal_k$ for the case that all of $p_{uv}$ are the same for $v\in \cbra{0,1}^n\backslash \cbra{u}$, and let it be $\bar{p}_u$ since usually $p_{uv}$ are close to each other for a given $v$. Then the probability that an error happens with output $\ket{u}$ equals $\sum_v p_{vu} = \sum_v \bar{p}_v \in [0,1]$. 
Typically, the probability of this error is nearly negligible and is bounded by a constant in an actual quantum device.
We leave the analysis of the general amplitude-damping noise to interested readers. Note that
\begin{align}
\sum_{
\substack{u, v\in \cbra{0,1}^n\\
u\ne v}
} E_{uv}^{\dagger} \ketbra{x}{x} E_{uv} + E_0^{\dagger} \ketbra{x}{x} E_0 &= p_{xx}\ketbra{x}{x} + \sum_{x'\in \cbra{0,1}^n\backslash \cbra{x}}p_{x'x}\ketbra{x'}{x'}\\
&=(1 - 2^n \bar p_{x})\ketbra{x}{x} + \bar p_{x}\sum_{x'\in\cbra{0,1}^n} \ketbra{x'}{x'}.
\label{eq:amplitude_damp_cal}
\end{align}
With the definitions of $\Bcal_k$ and $\Lambda_{\mathrm a}$ we have
\begin{align}
\Bcal_k &= \frac{(-i)^k}{2^n{n \choose k} }\sum_{x,S\in {[n]\choose k}}(-1)^{\sum_{j\in S}x_j}\pbra{\sum_{
\substack{u, v\in \cbra{0,1}^n\\
u\ne v}
} \tr{E_{uv}^{\dagger} \ketbra{x}{x} E_{uv} \gamma_{D(S)}} + \tr{E_0^{\dagger} \ketbra{x}{x} E_0\gamma_{D(S)}}}\\
&= \frac{(-i)^k}{2^n{n \choose k} }\sum_{x,S\in {[n]\choose k}}
(-1)^{\sum_{j\in S}x_j}
\pbra{(1-2^n\bar p_x) \tr{\ketbra{x}{x}\gamma_{D(S)}} + 
\bar p_{x}\sum_{x'\in\cbra{0,1}^n}  \tr{\ketbra{x'}{x'}\gamma_{D(S)}}
}
\label{eq:amplitude_damp_res}
\\
&= \frac{\sum_x(1-2^n\bar p_x)}{2^n} + \frac{(-i)^k}{2^n{n \choose k} }\sum_{x,S\in {[n]\choose k}}
(-1)^{\sum_{j\in S}x_j}\bar p_x \sum_{x'\in\cbra{0,1}^n}  (-i)^k (-1)^{\sum_{j\in S}x_j'}(-1)^k\\
&=1-\sum_x\bar p_x,
\end{align}
if $k\ne 0$, and $\Bcal_k = 1$ if $k =0$. Eq.~\eqref{eq:amplitude_damp_res} holds by the result in Eq.~\eqref{eq:amplitude_damp_cal}.
Since usually the noise strength $\sum_x \bar{p}_x$ in the device is close to zero, $\Bcal_k$ is close to one.
\item[(3)] For $X$-rotation noise with channel representation
\begin{align}
\Lambda_{\mathrm r}(\cdot) = R_{X}(\bm \theta)(\cdot)R_{X}(-\bm \theta)
\end{align}
 where $R_{X}(\bm \theta) = \exp(-i\sum_{l=1}^n \theta_l X_l/2)$ and $\theta_l\in [0,2\pi)$, we have
 \begin{align}
\Bcal_k = &\frac{(-i)^k}{2^n{n \choose k}} \sum_{S\in {[n]\choose k}} \sum_{x} (-1)^{\sum_{l\in S}x_l} \tr{\ketbra{x}{x}\Lambda_{\mathrm r}(\gamma_{D(S)})} \\
=&\frac{(-i)^k}{2^n{n \choose k}} \sum_{S\in {[n]\choose k}} \sum_{x} (-1)^{\sum_{l\in S}x_l} \tr{R_X(-{\bm \theta})\ketbra{x}{x}R_X({\bm \theta})\gamma_{D(S)}} \\
=&\frac{(-i)^k}{2^n{n \choose k}} \sum_{S\in {[n]\choose k}} \sum_x\prod_{l\in S}i(-1)^{x_l}\pbra{ \cos^2 \frac{\theta_l}{2} - \sin^2 \frac{\theta_l}{2}
}(-1)^{x_l}
\label{eq:B_k_rotation}\\
=& {n\choose k}^{-1}\sum_{S\in {[n]\choose k}}\prod_{l\in S}\cos\theta_l,
\end{align}
where Eq.~\eqref{eq:B_k_rotation} holds since $R_X(-\bm \theta)\ketbra{x}{x}R_X(\bm \theta)$ has the expansion form
\begin{equation}
R_X(-\bm \theta)\ketbra{x}{x}R_X(\bm \theta) = \bigotimes_{l=1}^n \begin{pmatrix}
\cos^2 \frac{\theta_l}{2} & -i \sin \frac{\theta_l}{2} \cos \frac{\theta_l}{2}\\
i \sin \frac{\theta_l}{2} \cos \frac{\theta_l}{2} & \sin^2 \frac{\theta_l}{2}
\end{pmatrix}
\end{equation}
with basis $\cbra{\ket{x_l}, \ket{1-x_l}}$, and $\gamma_{D(S)} = \prod_{l \in S} (iZ_l)$, where $Z$ is the Pauli-$Z$ gate. 
\item[(4)] For the noise with the representation being a Gaussian unitary channel $\Ucal_{Q}$, we have
\begin{align}
\Bcal_{k} 
&= \frac{(-i)^k}{2^n{n \choose k}} \sum_{S\in {[n]\choose k}} \sum_{x} (-1)^{\sum_{l\in S}x_l} \tr{\ketbra{x}{x}U_Q\gamma_{D(S)}U_Q^{\dagger}}\\
&=\frac{(-i)^k}{2^n{n \choose k}} \sum_{S,S'\in {[n]\choose k}} \sum_{x}  (-i)^k\det(Q|_{D(S'),D(S)}) (-1)^k\\
&= \frac{1}{{n \choose k}}\sum_{S,S'\in {[n]\choose k}}\det(Q|_{D(S'),D(S)}).
\end{align}
Below, we provide an illustration where selecting $Q$ from a signed permutation matrix demonstrates that $\Bcal_k$ has the potential to become zero.
Note that $Q$ is a signed permutation and hence det$(Q|_{D(S'),D(S)}) \ne 0$ iff $Q|_{D(S'),D(S)}$ maps $D(S')$ to $D(S)$.
Let the vector permutation representation of it be $\sigma_Q$, where $\sigma_{Qi}=jQ_{ij}$ if $Q_{ij}\ne 0$.
 Since $D(j)=(2j-1,2j)$, $\det(Q|_{D(S),D(S')}) \ne 0$ iff $\sigma_{Q}|_{D(S)}$ is a signed permutation of $D(S')$. Let $P\subset [k]$ be the largest set such that $\sigma_Q|_{D(P)}$ is a permutation for some $D(P')$ where size$(P')=\text{size}(P) =: r$. Consequently, $\Bcal_k = 0$ if $r<k<n$. 
For instance, take $\sigma_Q = (3, 4, 1, 5, 6, 8, 7, 2)$ with $n=4$. This leads to $\Bcal_k=0$ when $k\in \cbra{2, 3}$, since $r=1$ for both of these cases.

\end{itemize}

\section*{Supplementary note 4}
\label{app:property_noisyCFS}

Here we give proofs associated with the properties of the noisy classical fermionic channel.
\begin{proof}[Proof of Theorem~2 in the main file]
The noisy channel $\widetilde{\Mcal}$ can be represented as
\begin{equation}
\widetilde\Mcal = \int_Q d \mu(Q)\Ucal_{Q}^\dagger M_z \Lambda \Ucal_{Q} = \sum_{k = 0}^{n}f_{2k}\Pcal_{2k},
\end{equation}
since $\Pcal_{2k}$ forms $n+1$ irreducible representations for the unitary channel $\Ucal_{Q}$ (See details for the irreducible representations of a group in Supplementary Note 1). Then we have
\begin{equation}
\int_Q d \mu(Q) \Ebb_x \sbra{\supbraket{\bm 0|\Pcal_{2k} \Ucal_{Q}^\dagger | x}}
= \supbraket{\bm 0| \Pcal_{2k} \Mcal|\bm 0}
 = f_{2k} \supbraket{\bm 0|\Pcal_{2k} |\bm 0}.
\end{equation}
By using Lemma \ref{lem:p_2k_0_fidelity}, we can obtain $\supbraket{\bm 0|\Pcal_{2k} |\bm 0} = 2^{-n}{{n}\choose{k}}$. Combining this result with the previous findings leads us to
\begin{align}
   \Ebb[\hat{f}_{2k}] &= {2^{n}}{{n\choose k}}^{-1} \int_Q d \mu(Q) \Ebb_x \sbra{\supbraket{\bm 0|\Pcal_{2k} \Ucal_{Q}^\dagger | x}}\\
   &= {2^{n}}{{n\choose k}}^{-1} f_{2k} \supbraket{\bm 0|\Pcal_{2k} |\bm 0}\\
   &=f_{2k}.
\end{align}
\end{proof}

Below, we present proof of Lemma 1 in the main file.

\begin{proof}[Proof of Lemma 1 in the main file]
Recall that $\hat{f}_{2k} = 2^n {n \choose k}^{-1}\supbraket{\bm 0|\Pcal_{2k}\Ucal_{Q}^\dagger| x}$. Hence, by Lemma~\ref{lem:cal_maj_exp} (It can also be obtained by Proposition 1 in Ref.~\cite{wan2022matchgate}), we have
$\supbraket{\bm 0|\Pcal_{2k}\Ucal_{Q}^\dagger| x}$ is the coefficient of $x^k$ in the polynomial 
\begin{align}
p_{Q} (x)=\frac{1}{2^n} \pf\pbra{C_{\bm 0}}\pf\pbra{-C_{\bm 0}^{-1} + x(QC_{\ket{x}}Q^T)}.
\label{eq:polynomial_gaussian}
\end{align}
Hence $\hat{f}_{2k}$ is the coefficient of $x^k$ in the polynomial $2^n {n\choose k}^{-1} p_{Q}(x) = {n \choose k}^{-1}\pf\pbra{C_{\bm 0}}\pf\pbra{-C_{\bm 0}^{-1} + x(QC_{\ket{x}}Q^T)}$.
The polynomial in Eq.~\eqref{eq:polynomial_gaussian} can be calculated in $\Ocal(r^4)$ time since it has degree at most $r$, and the Pfaffian function of a $2n \times 2n$ matrix can be calculated in $\Ocal(n^3)$ time. It can be further optimized to $\Ocal(r^3)$ time by Appendix D of Ref. \cite{wan2022matchgate}, where $r$ is the rank of the covariance matrix.
\end{proof}

The variance can be bounded for the estimation with the bound of ${f_{2k}}$ and $\Ebb\sbra{f_{2k}^2}$.
In the following, we give proofs associated with the calculation of $\hat{f}_{2k}$. Since $\int_Q d \mu(Q)\Ebb_x\sbra{\supbraket{\bm 0|\Pcal_{2k}\Ucal_Q^\dagger | x}} = f_{2k} \supbraket{\bm 0|\Pcal_{2k}|\bm 0}$, we give the closed form for $\supbraket{\bm 0|\Pcal_{2k}|\bm 0}$ with the following lemma.
\begin{lemma}
\begin{align}
 \supbraket{\bm 0|\Pcal_{2k} |\bm 0} = 2^{-n}
\binom{n}{k}.
\label{eq:p_2k_0_fidelity}
 \end{align} 
\label{lem:p_2k_0_fidelity}
\end{lemma}
\begin{proof}
Recall that $\ketbra{\bm 0}{\bm 0} = \prod_{j=1}^n \frac{(\Ibb - i\gamma_{2j-1}\gamma_{2j})}{2}$. Hence, by substituting it into Eq.~\eqref{eq:p_2k_0_fidelity}, we obtain
\begin{align}
\supbraket{\bm 0|\Pcal_{2k} |\bm 0}&= \sum_{S\in {[2n]\choose 2k}}\abs{\tr{\ketbra{\bm 0}{\bm 0} \gamma_S}}^2\\
&= \sum_{S\in {[2n]\choose 2k}} \sum_{j = 0}^n \sum_{T\in {[n]\choose k} }\abs{\frac{(-i)^k}{4^n} \tr{\gamma_{D(T)}\gamma_S}}^2\\
&= \sum_{T\in {[n]\choose k}} \frac{1}{2^n} \\
&= \frac{{n\choose k}}{2^n},
\end{align} 
where $D(S)=\cbra{(2j-1,2j)|j\in S}$.
\end{proof}
As a generalization of Proposition 2 of Chen et al.~\cite{chen2021robust}, we give the bounds for $f_{2k}$ and $\Ebb\sbra{\hat f_{2k}^2}$, with prior knowledge of the average noise fidelity $\Bcal_k$ in the following two lemmas.
\begin{lemma}
\begin{align}
f_{2k} =  {2n \choose 2k}^{-1}
{n \choose k} \Bcal_k.
\end{align}
 \label{lem:abs_f_2k}
\end{lemma}

\begin{proof}
Since $f_{2k} = \Ebb\sbra{\hat f_{2k}}$, to find the value of $f_{2k}$, we can evaluate the expansion of $\hat{f}_{2k}$ as follows, 
\begin{align}
\Ebb\sbra{\hat{f}_{2k}} &= {2^{n}}{n \choose k}^{-1}\sum_{x}\int_Q d \mu(Q)\supbraket{\bm 0|\Pcal_{2k}\Ucal_{Q}^{\dagger} | x} \supbraket{x | \Lambda \Ucal_{Q} |\bm 0}  \\
&={2^{n}}{n \choose k}^{-1}\sum_{x} \int_Q d \mu(Q) \sbra{
\tr{\Ucal_{Q}^{\otimes 2} \supket{\bm 0}^{\otimes 2}\supbra{x}\Pcal_{2k}\supbra{x}\Lambda}
}
\label{eq:gamma_inv}
\\
&= {2^{n}}{n \choose k}^{-1} \sum_{x}\sum_{j=0}^{2n}\tr{\supket{\Rcal_{j}^{(2)}} \supbra{\Rcal_{j}^{(2)}}\cdot \supket{\bm 0}^{\otimes 2}\supbra{x} \Pcal_{2k}\supbra{x}\Lambda 
}
\label{eq:upper_bound_f}
\\
&={2^{n}}{n \choose k}^{-1} \sum_{x,0\leq j\leq 2n}
\supbra{\Rcal_{j}^{(2)}} \supket{\bm 0}^{\otimes 2} 
\pbra{\supbra{x}\Pcal_{2k} \otimes  \supbra{x} \Lambda}\supket{\Rcal_{j}^{(2)}},
\end{align}
where $\supket{\Rcal_{j}^{(2)}} = {2n \choose j}^{-1/2} \sum_{S\subseteq [2n],|S|=j} \supket{\gamma_S} \supket{\gamma_S}$. Eq.~\eqref{eq:gamma_inv} holds since $\supbraket{\bm 0|\Pcal_{2k}\Ucal_{Q}^{\dagger} | x} =  \supbraket{x|\Ucal_{Q} \Pcal_{2k}|\bm 0} = \supbraket{x|\Pcal_{2k}\Ucal_{Q} |\bm 0}$. Eq.~\eqref{eq:upper_bound_f} holds by Lemma \ref{lem:threemomentsMatchgate}. By the definition of $\Rcal_{j}^{(2)}$, we have
\begin{align}
\supbra{\Rcal_{j}^{(2)}} \supket{\bm 0}^{\otimes 2} &=  {2n \choose j}^{-1/2} \sum_{S\subseteq [2n],|S|=j} \supbraket{\gamma_S|\bm 0}^2\\
   &={2n \choose j}^{-1/2} {n \choose j/2} [j\text{ even}]\frac{(-1)^{j/2}}{2^n}   
\label{eq:upper_bound_f_fa}
\end{align}
and
\begin{align}
{(\supbra{x}\Pcal_{2k} \otimes  \supbra{x}\Lambda) \supket{\Rcal_{j}^{(2)}}} &= {2n \choose j}^{-1/2} \sum_{S\subseteq [2n],|S|=j} {\supbraket{x|\Pcal_{2k}|\gamma_S}
\supbraket{x|\Lambda |\gamma_S}}\\
&=[j=2k]{2n \choose 2k}^{-1/2} \sum_{S\in {[n]\choose k}}
\frac{i^k(-1)^{\sum_{l\in S}x_l}}{\sqrt{2^n}}
\supbraket{x|\Lambda |\gamma_{D(S)}}
\label{eq:upper_bound_f_fb}
\end{align}
where $D(S)=\cbra{2j-1, 2j|j\in S}.$
By substituting Eq.~\eqref{eq:upper_bound_f_fa}, and Eq.~\eqref{eq:upper_bound_f_fb} into Eq.~\eqref{eq:upper_bound_f}, we have
\begin{align}
f_{2k}& = (-i)^k{ 2n\choose 2k}^{-1} \sum_{s \in {[n]\choose k}} \sum_{x} (-1)^{\sum_{l\in S}x_l}\tr{\ketbra{x}{x}\Lambda(\gamma_{D(S)})}2^{-n}\\
&= { 2n\choose 2k}^{-1} {n \choose k} \Bcal_k,
\end{align}
where $\Bcal_k$ is defined in Eq. (5) in the main file.
\end{proof}

For the noise-free model, $f_{2k} =  {2n \choose 2k}^{-1}
{n \choose k}$.
We can use a similar technique to show that $\Ebb\sbra{\hat{f}_{2k}^{2}}$ is also bounded, as demonstrated in the following lemma.
\begin{lemma}
Let $\hat{f}_{2k} = {2^{n}}\supbraket{\bm 0|\Pcal_{2k} \Ucal_{Q}^\dagger | x}/{{n\choose k}}$ be an estimation of $f_{2k}$, then
\begin{align}
\Ebb\sbra{\abs{\hat{f}_{2k}}^2} ={n \choose k}^{-2} \sum_{0\leq l \leq k} {2n \choose 2l, 2k - 2l, 2l}_p^{-1} {n \choose l, k - l, l}_p^2 \Bcal_{2l}.
\end{align}
\label{lem:exp_fsquare_upper}
\end{lemma}
\begin{proof}
Firstly, we note that 
$\supbraket{\bm 0|\Pcal_{2k} \Ucal_{Q}^\dagger | x}^\ast = \sum_{S\in {[2n] \choose 2k}}\supbraket{\bm 0|\gamma_S^\dagger}\supbraket{\gamma_S^\dagger |\Ucal_{Q}^\dagger | x} = \supbraket{\bm 0|\Pcal_{2k} \Ucal_{Q}^\dagger | x}$ is a real number. Hence the estimation $\hat{f}_{2k}$ is a real number, and
\begin{align}
\Ebb\sbra{\abs{\hat{f}_{2k}}^2} &= {2^{2n}}{n \choose k}^{-2} \sum_{\bm x}\int_Q d \mu(Q) \sbra{\supbraket{\bm 0|\Pcal_{2k}\Ucal_{Q} |\bm x}^2 \supbraket{\bm x|\Lambda \Ucal_{Q} |\bm 0}}\\
&={2^{2n}}{n \choose k}^{-2} \sum_{\bm x}\int_Q d \mu(Q) \sbra{\tr{\Ucal_{Q}^{\otimes 3} 
\supket{\bm 0}^{\otimes 3} \pbra{\supbra{\bm x}\Pcal_{2k}
\otimes \supbra{\bm x}\Pcal_{2k}
\otimes \supbra{\bm x}\Lambda 
}
}}\\
&= {2^{2n}}{n \choose k}^{-2} \sum_x \sum_{\substack{k_1, k_2, k_3\geq 0\\
k_1 + k_2 + k_3 \leq 2n} }
\supbraket{\Rcal_{k_1, k_2, k_3} | \bm 0^{\otimes 3}} \pbra{\supbra{x}\Pcal_{2k}
\otimes \supbra{x}\Pcal_{2k}
\otimes \supbra{x}\Lambda 
} \supket{\Rcal_{k_1, k_2, k_3}}
\label{eq:square_f_thirdmoment}
\end{align}
where Eq.~\eqref{eq:square_f_thirdmoment} holds due to Lemma \ref{lem:threemomentsMatchgate}.
By the definition of $\Rcal_{k_1, k_2, k_3}$, we have
\begin{align}
\supbraket{\Rcal_{k_1, k_2, k_3} | \bm 0^{\otimes 3}} &= {2n \choose k_1,k_2,k_3}_p^{-1/2} \sum_{
\substack{S_1,S_2,S_3 \subseteq [2n]\text{ disjoint} \\
|S_i| = k_i, i\in [3]
}
}
\supbraket{\gamma_{S_1} \gamma_{S_2} |\bm 0}
\supbraket{\gamma_{S_2} \gamma_{S_3} |\bm 0}\supbraket{\gamma_{S_3} \gamma_{S_1} |\bm 0}\\
&= {2n \choose k_1,k_2,k_3}_p^{-1/2} [k_1,k_2,k_3 \text{ even}](-1)^{(k_1 + k_2 + k_3)/2} 2^{-3n/2}{n \choose \frac{k_1}{2}, \frac{k_2}{2}, \frac{k_3}{2}}_p
\label{eq:f_square_first_proof}
\end{align}
where ${ n\choose k_1, k_2, k_3}_p = {n \choose k_1,k_2, k_3, n - k_1-k_2 - k_3}$, and
\begin{align}
&\pbra{\supbra{x}\Pcal_{2k}
\otimes \supbra{x}\Pcal_{2k}
\otimes \supbra{x}\Lambda 
} \supket{\Rcal_{k_1, k_2, k_3}} \nonumber\\
&={2n \choose k_1,k_2,k_3}_p^{-1/2} \sum_{
\substack{
S_1,S_2,S_3 \subseteq [2n]\text{ disjoint}\\
\abs{S_i}=k_i,i \in [3]
}
}\supbraket{x|\Pcal_{2k}|\gamma_{S_1}\gamma_{S_2}} \supbraket{x|\Pcal_{2k}|\gamma_{S_2}\gamma_{S_3}} \supbraket{x|\Lambda|\gamma_{S_3}\gamma_{S_1}} 
\\
&= {2n \choose k_1,k_2,k_3}_p^{-1/2} [k_1=k_3\text{ even},k_2 = 2k-k_1] \sum_{
\substack{S_1,S_2,S_3\subseteq [n]\text{ disjoint}\\
\abs{S_1}=\abs{S_3}=\frac{k_1}{2},\abs{S_2} = k-\frac{k_1}{2}
}
} 
\frac{i^{\abs{S_1} + 2\abs{S_2} + \abs{S_3}} (-1)^{\sum_{l\in S_1 \cup S_3} x_l}\supbraket{x|\Lambda|\gamma_{D(S_3\cup S_1)}}}{2^n} 
\\
&= {2n \choose 2l,2k - 2l,2l}_p^{-1/2}[k_1=k_3=2l,k_2=2k-2l] \nonumber\\
&\qquad \times
\sum_{\substack{
S_1,S_3\subseteq [n] \text{ disjoint}\\
\abs{S_1} = \abs{S_3}=l
}
}\frac{(-1)^{k}}{2^{3n/2} } {n-2l \choose k-l}\tr{\gamma_{D(S_1\cup S_3)}\Lambda|(\ketbra{x}{x})}^\ast(-1)^{\sum_{l\in S_1\cup S_3} x_l}\\
&= {2n \choose 2l,2k - 2l,2l}_p^{-1/2}[k_1=k_3=2l,k_2=2k-2l] \sum_{S\in {[n]\choose 2l}
}\frac{(-1)^{k}}{2^{3n/2} } {n-2l \choose k-l}{2l \choose l}\tr{\ketbra{x}{x}\Lambda(\gamma_{D(S)})}(-1)^{\sum_{l\in S} x_l}.
\label{eq:f_square_sec_proof}
\end{align} 
Substituting Eq.~\eqref{eq:f_square_first_proof} and Eq.~\eqref{eq:f_square_sec_proof} into Eq.~\eqref{eq:square_f_thirdmoment} gives us the following equality
\begin{align}
\Ebb\sbra{\hat{f}_{2k}^2} = {n \choose k}^{-2} \sum_{0\leq l \leq \min(k,n-k)} {2n \choose 2l, 2k - 2l, 2l}_p^{-1} {n \choose l, k - l, l}_p^2 \Bcal_{2l},
\end{align}
where $\Bcal_{l}$ is defined in Eq. (5) in the main file.
\end{proof}

 In Supplementary Note 6, we will give the upper bounds for the number of samples $R_e, R_c$ to bound estimation and calibration errors $\varepsilon_e, \varepsilon_c$ respectively, based on the results of the above two lemmas.

\section*{Supplementary note 5}
\label{app:cal_est_quantities}
In this section, we analyze the simplification of the estimations of several physical quantities with error-mitigated matchgate shadow.

\begin{lemma}
Given $Q\in \mathrm{Orth}(2n)$, a Majorana operator $\gamma_S$ and computational basis state $\ket{x}, $ where $x\in \cbra{0,1}^n$, we have
\begin{align}
    \tr{U_Q^\dagger \gamma_S U_Q \ketbra{x}{x}} = i^{|S|/2}\pf(QC_{\ket{x}}Q^T|_S),
\end{align}
if $|S|$ is even, and zero otherwise.
\label{lem:cal_maj_exp}
\end{lemma}
\begin{proof}
It is easy to check that if $|S|$ is odd, then $\tr{U_Q^\dagger \gamma_S U_Q \ketbra{x}{x}} = 0$. In the following, we assume that $|S|$ is even. With the definitions of $C_{\ket{x}}$ and the {Pfaffian} function, we have
$\tr{\gamma_S\ketbra{x}{x}} = i^{|S|/2} \pf(C_{\ket{x}}|_S)$.
By the properties of the matchgate operation on the Majorana operator, we have
\begin{align}
\tr{U_Q^{\dagger} \gamma_S U_Q \ketbra{x}{x}} &= \sum_{S'\in {[n]\choose |S|}} \text{det}(Q|_{SS'})\tr{\gamma_{S'} \ketbra{x}{x}}\\
&= i^{|S|/2}\pf(QC_{\ket{x}}Q^T|_{S}),
\end{align}
where the last equality holds since $\det(A)\pf(B) = \pf(ABA^T)$ for any  $2n\times 2n$ matrix $A$ and a skew-symmetric matrix $B$ (A \textit{skew-symmetric} matrix $B$ is one that satisfies $B^T = -B$)~\cite{bravyi2017complexity}.
\end{proof}
Let $\hat{v}$ be the estimation result obtained by computing $\tr{\widetilde{\gamma}_S \widehat{\Mcal}^{-1} \pbra{U_Q^\dagger \ketbra{x}{x} U_Q}}$ where $\widetilde\gamma_S = U_{Q_1}^\dagger \gamma_S U_{Q_1}$.
We can simplify the estimate of the expectation value of the Majorana operator $\widetilde{\gamma}_S$ as follows:
\begin{align}
\hat{v} &= \tr{\widetilde{\gamma}_S{\widehat\Mcal^{-1}(U_Q^\dagger \ketbra{x}{x}U_Q)}}\\
&= \tr{{\gamma}_S{\widehat\Mcal^{-1}(U_{Q_1}U_Q^\dagger \ketbra{x}{x}U_QU_{Q_1}^{\dagger})}}\\
&= \frac{1}{\hat{f}_{|S|}2^n}\sum_{S'\in {[2n]\choose |S|}}\tr{{\gamma}_S \gamma_{S'}}\tr{\gamma_{S'}U_{Q_1}U_Q^\dagger \ketbra{x}{x}U_QU_{Q_1}^{\dagger}}\\
   &= \frac{i^{|S|/2}}{\hat{f}_{|S|}}\pf(Q_1Q^TC_{\ket{x}} QQ_1^T|_S),
\end{align}
where $\widetilde{\gamma}_j = \sum_{k=1}^{2n} \det(Q_1(j,k))\gamma_{k}$, and $Q_1(j,k)$ is the $(j,k)$-th element of matrix $Q_1$.


 Wan et al.~\cite{wan2022matchgate} introduce techniques to determine the overlap between a Gaussian state and any quantum state, along with a $\tau$-slater determinant and any pure state. Notably, the distinction between the noisy channel and the original CS channel boils down to the varying coefficients $f_{2k}$. To ensure the integrity of our paper, we outline the equations used to calculate these quantities below. One can find detailed information about the deviation in Ref.~\cite{wan2022matchgate}.

(1) Calculate  the overlap between an $n$-qubit Gaussian state with density matrix $\rho_g$ and any $n$-qubit quantum state with density matrix $\rho$, denoted as $\tr{\rho \rho_g}$. By letting $H = \rho_g$, the estimator can be calculated efficiently since $\tr{\rho_g \Pcal_{2k}\pbra{U_Q^\dagger \ketbra{x}{x}U_Q}}$ is the coefficient of $z^k$ in the polynomial $p_{\rho_g, \chi} = \frac{1}{2^n} \pf(C_{{\rho_g}})\pf(-C_{{\rho_g}}^{-1} + z C_{\chi})$, where $\chi = U_Q^\dagger \ketbra{x}{x}U_Q$. Hence $\hat{v}$ can be estimated in polynomial time with the polynomial interpolation method.

(2) Calculate the inner product between $\tau$-fermionic Slater determinants $\ket{\phi_\tau}$ and any quantum pure state $\ket{\psi}$, denoted as $\braket{\psi|\phi_\tau}$ and since the shadow channel $\widetilde\Mcal$ twirls any observable to the even subspace, here we assume $\ket{\psi}$ is in the odd subspace. We require $n+1$ qubits if $\tau$ is even (and $n+2$ qubits otherwise) and suppose $\ket{\psi} = U_{\psi}\ket{\bm 0}$. By performing Hadamard gate on the first qubit, followed by control-$U_{\psi}$ gate, we generate the initial state {$\rho = \frac{(\ket{0^{n+1}} + \ket{1}\ket{\psi})(\bra{0^{n+1}} + \bra{1}\bra{\psi})}{2}$}, and let observable $H = \ket{1}\ket{\phi_\tau}\bra{0^{n+1}}$. It is easy to check {$\tr{\rho H} = \braket{\psi|\phi_\tau}/2$}.
The quantity $\tr{HU_Q^\dagger \ketbra{x}{x}U_Q}$ is the coefficient of $z^k$ in the polynomial $p_{\phi,\chi} = \frac{i^{\tau/2}}{2^{n-\tau/2}}\pf\pbra{C_{\bm 0} + z Q'^T \widetilde Q C_{\chi} \widetilde{Q}^T Q'|E_{\tau}}$, where $C_{\bm \rho}$ is the covariance matrix for $\rho$, $Q'=\bigoplus_{j\in [\tau]} \frac{1}{\sqrt{2}} \begin{pmatrix}
    1 & -i \\
    1 &i
\end{pmatrix}\bigoplus_{j=\tau + 1}^n \begin{pmatrix}
    1 & 0 \\
    0 &1
\end{pmatrix}$ and $E_{\tau} = [2n]\backslash \cbra{1,3,\ldots, 2\tau - 1}$, and we can get all of the coefficients in $\Ord{n^4}$ time, hence the estimator $\hat{v}$ can be obtained efficiently.

\section*{Supplementary note 6}
\label{app:err_estimation}

To limit the estimation error, we first provide the \textbf{MedianOfMeans} method~\cite{huang2020predicting}, as shown in Algorithm \ref{alg:Median_of_means}. To give a better bound, the number of repetitions is usually set as $K = \Ocal(\Var{\hat{v}}/\varepsilon^2)$ and $ N = \Ocal(\ln {\delta^{-1}})$ in the algorithm, where $\hat{v}$ is the estimation in a single round.

\begin{algorithm}
\caption{Median of Means}
\label{alg:Median_of_means}
\begin{algorithmic}
\Input{ The number of repetitions $N,K$, and estimations $\cbra{v_j}_{j=1}^{R}$, where $R = NK$.}
\EndInput
\Output{ Estimation $\bar v$.}
\EndOutput
\For{$l \leftarrow 1$ to $K$}
 \State $\bar{v}_l:= \frac{1}{N} \sum_{j = (l-1)N+ 1}^{KN} v_j$;
\EndFor
\State $\bar{v} = \text{median}(v_1,\ldots, v_K)$;
\State \textbf{Return}{ $\bar{v}$}
\end{algorithmic}
\end{algorithm}

\begin{lemma}[Theorem S1 in Ref.~\cite{huang2020predicting}]
Let $S=\cbra{\hat v_1,\ldots, \hat v_R}$ be $R=NK$ identical and independent samples from the same distribution with $N= 34\Var{\hat v}/\varepsilon^2$ for any $\hat v\in S$, and $K = 2\ln (2\delta^{-1})$ for some precision parameters $\varepsilon,\delta \in \sbra{0,1}$, then the estimation $\bar{v}$ in Algorithm \ref{alg:Median_of_means} satisfies
\begin{align}
\Pr[\abs{\bar{v} - \Ebb\sbra{\hat v}}\geq \varepsilon ]\leq \delta.
\end{align}
\label{lem:Median_of_means}
\end{lemma}

The required number of {samples for accurate estimation} can be upper bounded by the variance of 
  estimation $\hat{v}_i^{(j)}$ in a single round {of sampling} 
 (as shown in Step 14 of Algorithm 1 in the main file)  by utilizing \textbf{MedianOfMeans} method~\cite{huang2020predicting}.
We give the variance of a single estimation $\hat{v}$
in the following lemma.
\begin{lemma}
For any observable $H$ and an unknown quantum state $\rho$, the estimation $\hat{v} $ generated in Step 16 of Algorithm 1 in the main file for quantity $\tr{H\widehat{\Mcal}^{-1}\widetilde{\Mcal}(\rho)}$  satisfies
\begin{align}
\begin{aligned}
\Var{\hat{v} } &\leq 
\frac{(1+\varepsilon_c)^2}{2^{2n}}\sum_{0\leq l_1 + l_2 + l_3\leq n} \frac{(-1)^{l_1+l_2+l_3}{n\choose l_1, l_2, l_3}_p{2n\choose 2l_1 + 2l_2} {2n\choose 2l_2 + 2l_3}\Bcal_{l_1 + l_3}}{{2n\choose 2l_1, 2l_2, 2l_3}_p{n \choose l_1+l_2}{n \choose l_2 + l_3}\Bcal_{l_1 + l_2}\Bcal_{l_2+l_3}} \\
&\quad \times \sum_{
\substack{S_1,S_2,S_3\ \mathrm{ disjoint}\\
|S_i| = 2l_i,i\in [3]}
}\tr{\widetilde\gamma_{S_1} \widetilde\gamma_{S_2}H_0} \tr{\widetilde\gamma_{S_2} \widetilde\gamma_{S_3} H_0^\dagger} \tr{\widetilde\gamma_{S_3}\widetilde\gamma_{S_1} \rho },  
\end{aligned}
\label{eq:var_dev_Single}
\end{align}
where $H_0 = H -\tr{H}\frac{\Ibb}{2^n}$ is the traceless part of $H$.
\label{lem:variance_estimation}
\end{lemma}
\begin{proof}
By the definition of $\hat{v}$, we have
\begin{align}
\Var{\hat{v}} &= \Ebb\sbra{\abs{\hat{v}  - \Ebb\sbra{\hat{v} }}^2}\\
&= \Ebb_{x}\int_{Q} d\mu(Q) \sbra{\abs{\supbraket{H_0^{\dagger}|\widehat\Mcal^{-1}\Ucal_{Q}^{\dagger}| x} - \supbraket{H_0^\dagger|\widehat\Mcal^{-1}\widetilde\Mcal|\rho}}^2}\\
&\leq \sum_{x}\int_{Q} d\mu(Q) {\abs{\supbraket{H_0^{\dagger}|\widehat\Mcal^{-1}\Ucal_{Q}^{\dagger}| x}}^2 \supbraket{x|\Lambda \Ucal_{Q} |\rho}}\\
&=\sum_{x} \int_{Q} d\mu(Q) {\tr{\Ucal_{Q}^{\otimes 3} \widehat\Mcal^{-1}\supket{H_0} \widehat\Mcal^{-1}\supket{H_0^{\dagger}} \supket{\rho} \supbra{x}^{\otimes 2} \supbra{x}\Lambda  }}
\label{eq:third_moment_var_v}\\
&= \sum_{x} \sum_{0\leq k_1+k_2+k_3\leq 2n}
\supbra{\widetilde\Rcal_{k_1k_2k_3} }\widehat\Mcal^{-1} \supket{H_0}\widehat\Mcal^{-1} \supket{H_0^{\dagger}}\supket{\rho} \supbra{x}\supbra{x}\supbra{x}\Lambda \supket{\widetilde\Rcal_{k_1k_2k_3}},
\end{align} 
where $H_0 = H - \tr{H}\frac{\Ibb}{2^n}$ is the traceless part of $H$. Eq.~\eqref{eq:third_moment_var_v} holds since $\supbraket{H_0^{\dagger}|\widehat{\Mcal}^{-1}\Ucal_{Q}^{\dagger}|x}= \supbraket{x|\Ucal_{Q}\widehat{\Mcal}^{-1}|H_0}$. With the definition of $\widetilde\Rcal_{k_1k_2k_3}$ in Eq.~\eqref{eq:third_moment} and the definition of $\widehat{\Mcal}$, we have
\begin{align}
&\supbra{\Rcal_{k_1k_2k_3}} \widehat\Mcal^{-1} \supket{H_0}\widehat\Mcal^{-1} \supket{H_0^\dagger}\supket{\rho} \nonumber\\
&= {2n\choose k_1, k_2, k_3}_p^{-1/2}\sum_{
\substack{S_1, S_2, S_3 \text{ disjoint}\\
|S_i| = k_i,i\in [3]\\
(k_i+k_{i+1}) \text{ even}, i\in [2]}
}
2^{-3n/2} \hat f_{k_1 + k_2}^{-1} \hat f_{k_2 + k_3}^{-1} \tr{\widetilde\gamma_{S_2}^\dagger \widetilde\gamma_{S_1}^\dagger H_0} \tr{\widetilde\gamma_{S_3}^\dagger \widetilde\gamma_{S_2}^\dagger H_0^\dagger} \tr{\widetilde\gamma_{S_1}^\dagger\widetilde\gamma_{S_3}^\dagger \rho } \\
&\leq \frac{2^{-3n/2}}{{2n\choose k_1, k_2, k_3}_p^{\frac{1}{2}}}\sum_{
\substack{S_1, S_2, S_3 \text{ disjoint}\\
|S_i| = k_i,i\in [3]\\
(k_i+k_{i+1}) \text{ even}, i\in [2]}
}\frac{(1+\varepsilon_c)^2}{f_{k_1 + k_2} f_{k_2 + k_3}} \tr{\widetilde\gamma_{S_1}\widetilde\gamma_{S_2}  H_0} \tr{\widetilde\gamma_{S_2} \widetilde\gamma_{S_3} H_0^\dagger} \tr{\widetilde\gamma_{S_3}\widetilde\gamma_{S_1} \rho }
 \label{eq:f_b_substitute}\\
 &=\frac{2^{-3n/2}{2n\choose k_1+k_2}{2n\choose k_2+k_3}(1+\varepsilon_c)^2}{{2n\choose k_1, k_2, k_3}_p^{\frac{1}{2}}{n\choose \frac{k_1+k_2}{2}}{n\choose \frac{k_2+k_3}{2}}\Bcal_{\frac{k_1 + k_2}{2}}\Bcal_{\frac{k_2 + k_3}{2}}}\sum_{
\substack{S_1, S_2, S_3 \text{ disjoint}\\
|S_i| = k_i,i\in [3]\\
(k_i+k_{i+1}) \text{ even}, i\in [2]
}
}
 \tr{\widetilde\gamma_{S_1}\widetilde\gamma_{S_2}  H_0} \tr{\widetilde\gamma_{S_2}\widetilde\gamma_{S_3}  H_0^\dagger} \tr{\widetilde\gamma_{S_3}\widetilde\gamma_{S_1} \rho },
 \label{eq:f_cali_bound_substitute}
\end{align}
where  Eq.~\eqref{eq:f_b_substitute} holds since $\frac{f_{2k}}{\hat{f}_{2k}}\leq 1+ \varepsilon_c$ for any $0\leq k\leq n$ from Lemma \ref{lem:calibrate_error_bound}, and Eq.~\eqref{eq:f_cali_bound_substitute} holds by Lemma~\ref{lem:abs_f_2k}.
Furthermore,
\begin{align}
&\supbra{x}^{\otimes 2}\supbra{x}\Lambda \supket{\Rcal_{k_1k_2k_3}} =2^{-n}{2n\choose 2l_1, 2l_2, 2l_3}_p^{-1/2}\sum_{\substack{S_1\in {[n] \choose l_1}\\
S_2\in {[n]\backslash S_1\choose l_2}\\
S_3 \in {[n]\backslash S_1 \cup S_2 \choose l_3}}
}
(-i)^{l_1 + 2l_2+ l_3} (-1)^{l_1+l_3}(-1)^{\sum_{j\in S_1 \cup S_3} x_j} \supbraket{x|\Lambda | \widetilde\gamma_{D(S_1\cup S_3)}}\\
&={2n\choose 2l_1, 2l_2, 2l_3}_p^{-1/2} 
\frac{(-i)^{l_1 + 2l_2+ l_3}}{2^{3n/2}}
{n - l_1 - l_3 \choose l_2}\sum_{s\in {[n] \choose l_1 + l_3}} (-1)^{\sum_{j\in S} x_j}{l_1+ l_3 \choose l_1} \tr{\ketbra{x}{x} \Lambda(\widetilde\gamma_S)}\\
&={2n\choose 2l_1, 2l_2, 2l_3}_p^{-1/2} 
\frac{(-1)^{l_1+l_2+l_3}}{2^{n/2}}
{n \choose l_1,l_2,l_3}_p
\Bcal_{l_1 + l_3},
\end{align}
where $l_i = k_i/2$ for $i\in [3]$.
By combining all the elements, we arrive at
\begin{align}
\begin{aligned}
\Var{\hat{v} } &\leq 
\frac{(1+\varepsilon_c)^2}{2^{2n}}\sum_{0\leq l_1 + l_2 + l_3\leq n} \frac{(-1)^{l_1+l_2+l_3}{n\choose l_1, l_2, l_3}_p{2n\choose 2l_1 + 2l_2} {2n\choose 2l_2 + 2l_3}\Bcal_{l_1 + l_3}}{{2n\choose 2l_1, 2l_2, 2l_3}_p{n \choose l_1+l_2}{n \choose l_2 + l_3}\Bcal_{l_1 + l_2}\Bcal_{l_2+l_3}} \\
&\quad \times \sum_{
\substack{S_1,S_2,S_3\text{ disjoint}\\
|S_i| = 2l_i,i\in [3]}
}\tr{\widetilde\gamma_{S_1} \widetilde\gamma_{S_2}H_0} \tr{\widetilde\gamma_{S_2} \widetilde\gamma_{S_3} H_0^\dagger} \tr{\widetilde\gamma_{S_3}\widetilde\gamma_{S_1} \rho}.  
\end{aligned}
\label{eq:var_dev_S}
\end{align}
\end{proof}

It's worth noting that in the absence of noise, the variance bound is equivalent to the variance described in Ref.~\cite{wan2022matchgate}, where $\Bcal_k=1$ for any $k$. Using Chebyshev's inequality, we can ensure that the error $\varepsilon$ of $\hat{v}$ is limited to $\frac{\Var{\hat{v}}}{\varepsilon^2}$ with a high degree of certainty. To assess the effectiveness of this bound, we provide a specific bound on the variance for estimating various physical quantities:
\begin{itemize}
     \item [(0)] Calculate $\tr{\rho \widetilde\gamma_S}$, where $\abs{S} = 2k$. If we choose the observable $H = \widetilde\gamma_S$, then the variance can be bounded by
\begin{align}
\frac{(1+\varepsilon_c)^2 {2n \choose 2k}}{\Bcal_{k}^2{n \choose k}} =   \Ord{ \Bcal_{k}^{-2}n^{k}}.
\end{align}
This can be further simplified to $\Ord{n^{k}}$ for noises with constant average fidelity in $\Gamma_{2k}$ subspace.
 \item[(1)] Calculate $\tr{\rho {\rho_g}}$, where ${\rho_g}$ is a Gaussian state. By choosing the observable $H = {\rho_g}$, the variance is bounded to $\Ord{\frac{\Bcal_{\max}\sqrt n}{\Bcal_{\min}^2}}$, where $\Bcal_{\max} = \max_k \abs{\Bcal_k}$ and $\Bcal_{\min} = \min_k \abs{\Bcal_k}$. 
It can be further simplified to $\Ord{\sqrt n}$ if the average fidelities of noise $\Bcal_k$ are constants for any $k\in \cbra{0,\ldots, n}$. 
\item[(2)] Calculate $\braket{\psi|\phi_\tau}$, where $\ket{\phi_\tau}$ is a $\tau$-fermionic Slater determinant. Let the initial input state of the quantum device be $\rho = \frac{(\ket{\bm 0} + \ket{\psi})(\bra{\bm 0} + \bra{\psi})}{2}$, and observable $H = \ket{\phi_\tau}\bra{\bm 0}$, the variance is bounded to $\Ord{\frac{\sqrt{n}\ln n \Bcal_{\max}\ln(m/\delta_e)}{\Bcal_{\min}^2\varepsilon_e^2}}$, which equals $\Ord{\frac{\sqrt{n}\ln n \ln(m/\delta_e)}{\varepsilon_e^2}}$ if the average fidelities of noise $\Bcal_k$ are constants for any $k\in \cbra{0,\ldots, n}$.
 \end{itemize}
We give the details for the simplification of variance in Supplementary Note 8. 
With this lemma and \textbf{MedianOfMeans} method, we can obtain Theorem~\ref{thm:est_samples}.

\begin{theorem}
Given a quantum state $\rho$ and observables $\cbra{H_i}_{i = 1}^m$.
With the number of estimation samplings
\begin{align}
R_e &= \frac{68(1+\varepsilon_c)^2\ln (2m/\delta_e)}{\varepsilon_e^2} 2^{-2n}\sum_{0\leq l_1 + l_2 + l_3\leq n} \frac{(-1)^{l_1+l_2+l_3}{n\choose l_1, l_2, l_3}_p{2n\choose 2l_1 + 2l_2} {2n\choose 2l_2 + 2l_3}\Bcal_{l_1 + l_3}}{{2n\choose 2l_1, 2l_2, 2l_3}_p{n \choose l_1+l_2}{n \choose l_2 + l_3}\Bcal_{l_1 + l_2}\Bcal_{l_2+l_3}}\nonumber\\
& \qquad \times \sum_{
\substack{S_1,S_2,S_3 \ \mathrm{ disjoint}\\
|S_i| = 2l_i,i\in [3]}
} \tr{\widetilde\gamma_{S_1} \widetilde\gamma_{S_2}H_0} \tr{\widetilde\gamma_{S_2} \widetilde\gamma_{S_3} H_0} \tr{\widetilde\gamma_{S_3}\widetilde\gamma_{S_1} \rho},
\end{align}
 where $H_0 = \max_i (H_i -\tr{H_i}\frac{\Ibb}{2^n})$, by running Algorithm~1 in the main file, we can bound the estimation error $\abs{\hat{v}_i - \tr{H_i\widehat{\Mcal}^{-1} \widetilde{\Mcal} (\rho)}}\leq \varepsilon_e$ for any $i\in [m]$ with success probability $1 - \delta_e$.
\label{thm:est_samples}
\end{theorem}
\begin{proof}
The theorem can be obtained from Lemma~\ref{lem:Median_of_means} with 
 the number of copies for $\rho$ being $R_e=K_eN_e$, where $K_e = 2\ln (2m\delta^{-1}_e)$,  and $N_e = \frac{34 \Var{\hat v}}{\varepsilon_e^2}$. With the bound on the variance in Lemma~\ref{lem:variance_estimation}, the value of $\abs{f_{2k}}$ provided in Lemma~\ref{lem:abs_f_2k} and the union bound, we see that with the number of estimation samplings
 \begin{align}
R_e &=N_eK_e\\
&= \frac{68 \ln (2m/\delta_e)\Var{\hat v}}{\varepsilon_e^2}\\
&\leq \frac{68(1+\varepsilon_c)^2\ln (2m/\delta_e)}{\varepsilon_e^2 2^{2n}} \sum_{0\leq l_1 + l_2 + l_3\leq n} \frac{(-1)^{l_1+l_2+l_3}{n\choose l_1, l_2, l_3}_p{2n\choose 2l_1 + 2l_2} {2n\choose 2l_2 + 2l_3} \Bcal_{l_1 + l_2}}{{2n\choose 2l_1, 2l_2, 2l_3}_p{n \choose l_1+l_2}{n \choose l_2 + l_3} \Bcal_{l_2+l_3}\Bcal_{l_1 + l_3}} \nonumber\\
&\qquad \times \sum_{
\substack{S_1,S_2,S_3\text{ disjoint}\\
|S_i| = 2l_i,i\in [3]}
} \tr{\widetilde\gamma_{S_1} \widetilde\gamma_{S_2}H_0} \tr{\widetilde\gamma_{S_2} \widetilde\gamma_{S_3} H_0} \tr{\widetilde\gamma_{S_3}\widetilde\gamma_{S_1} \rho},
 \end{align}
 where $\abs{\hat{v}_i - \tr{H_i\widehat{\Mcal}^{-1} \widetilde{\Mcal} (\rho)}}\leq \varepsilon_e$ for any $i\in [m]$, with success probability $1-\delta_e$.
\end{proof}

From this theorem, we can give the estimation of the number of samplings to approximate $k$-RDM, the overlap of $m$ number of Gaussian state and any quantum state, and $m$ number of $\tau$-fermionic Slater determinant.  The number of estimation samplings equals (1) $R_e = \Ord{\frac{k\ln(n/\delta_e)}{\varepsilon_e^2\Bcal_{k}^2}n^{k}}$ to estimate $k$-RDMs; (2) $R_e = \Ord{\frac{\sqrt{n}\Bcal_{\max}\ln(m/\delta_e)}{\Bcal_{\min}^2\varepsilon_e^2}}$ 
to estimate $\cbra{\tr{\rho {\rho_{g_j}}}}_{j = 1}^m$; and
(3) $R_e = \Ocal\pbra{\frac{\sqrt{n}\ln n \Bcal_{\max}\ln(m/\delta_e)}{\Bcal_{\min}^2\varepsilon_e^2}}$
to estimate $\cbra{\braket{\psi|\phi_{\tau_j}}}_{j=1}^m$. The number of samples in all of these applications is polynomial to the number of qubits $n$. 
We give more detailed calculations for the required number of samples in Supplementary Note 8.

\section*{Supplementary note 7}
\label{app:err_calibration}

As an extension of the calibration bound introduced in Ref.~\cite{chen2021robust}, the subsequent lemma illustrates the minimum calibration error resulting from the estimation of $f_{2k}$.

\begin{lemma}
Consider $\widehat{\Mcal}$ as the estimated channel for $\widetilde{\Mcal}$, where $\widehat{\Mcal} = \sum_{k = 0}^n \hat{f}_{2k}\Pcal_{2k}$. Let $\rho$ represent a quantum state, and let the observable $H$ belong to the subspace $\Gamma_{\mathrm{even}}$, then
\begin{align}    \left|\tr{H\widehat{\Mcal}^{-1} \widetilde{\Mcal}(\rho)} - \tr{H\rho}\right|=\vabs{H}_{\infty}\max_{k} (\hat{f}_{2k}^{-1}f_{2k} - 1)\leq \varepsilon_c \vabs{H}_{\infty}
\end{align}
if $\max_{k}{\abs{\hat f_{2k}-f_{2k}}}\leq \frac{\abs{f_{2k}}\varepsilon_c}{1 + \varepsilon_c}$.
\label{lem:calibrate_error_bound}
\end{lemma}
\begin{proof}
For the quantum state in the even subspace $\Gamma_{\mathrm{even}}$, we have
\begin{align}
\left|\tr{H\widehat{\Mcal}^{-1} \widetilde{\Mcal}(\rho)} - \tr{H\rho}\right|&= \abs{\tr{\rho\sum_{k}\pbra{\hat f_{2k}^{-1}f_{2k}
 - 1
}\Pcal_{2k}(H)} }\\
&\leq {\abs{\tr{H \rho}}\max_k \abs{\hat{f}_{2k}^{-1}f_{2k} - 1}}\\
&\leq \vabs{H}_{\infty}\max_k \abs{\hat{f}_{2k}^{-1}f_{2k} - 1}\\
&\leq \varepsilon_c \vabs{H}_{\infty}.
\end{align}
Here $\vabs{\cdot}_{\infty}$ denotes the spectral norm.  Suppose ${\abs{\hat f_{2k}-f_{2k}}}\leq \epsilon_{f_{2k}}$, then
\begin{align}
 \abs{\hat{f}_{2k}^{-1}f_{2k} - 1} &\leq  \frac{\abs{f_{2k} - \hat{f}_{2k}}}{f_{2k}-\abs{f_{2k} - \hat{f}_{2k}}} \\
&=\frac{\epsilon_{f_{2k}}}{\abs{f_{2k}} - \epsilon_{f_{2k}}}\\
&\leq \varepsilon_c.
\end{align}
Hence,
$\epsilon_{f_{2k}}\leq \frac{\abs{f_{2k}}\varepsilon_c}{1 + \varepsilon_c}$.
\end{proof}

With this lemma, we can further give the required number of samplings to bound the error resulting from the estimation of the noisy channel $\widehat{\Mcal}$.

\begin{theorem}
For any given unknown quantum state $\rho$ and observables $\cbra{H_i}_{i=1}^m$, with the number of calibration samplings 
\begin{equation}
R_c =\Ord{\frac{\Bcal_{\max}\sqrt{n} \ln n\ln (1/\delta_c)}{\Bcal_{\min}^2\varepsilon_c^2}},
\label{eq:R_c_fullEq}
\end{equation}
where $\Bcal_{\max }= \max_k \abs{\Bcal_k}$ and $\Bcal_{\min} = \min_k \abs{\Bcal_k}$. 
we can bound the error resulting in the estimation of channel $\widetilde{\Mcal}$ to 
\begin{equation}
    \abs{\tr{H_i\widehat{\Mcal}^{-1}\widetilde{\Mcal} (\rho) } - \tr{H_i\rho}}\leq \varepsilon_c \vabs{H_0}_\infty
    \label{eq:calibrate_samples}
\end{equation}
for any $i\in [m]$ with success probability $1 - \delta_c$ by running the calibration process in Algorithm 1 in the main file, where $H_0$ is the noiseless term of the observable $H_i$.
\label{thm:calibration_samples}
\end{theorem}
By Lemma \ref{lem:calibrate_error_bound} and error bound of \textbf{MedianOfMeans}, we can get
\begin{equation}
    R_c = \frac{68 (1 + \varepsilon_c)^2}{\varepsilon_c}\max_k \frac{\Ebb\sbra{\hat{f}_{2k}^2}}{\abs{f_{2k}}^{2}}.
\end{equation}
By substituting the results of $\Ebb\sbra{\hat f_{2k}^2}, \abs{f_{2k}}$ in Lemma \ref{lem:abs_f_2k} and Lemma \ref{lem:exp_fsquare_upper} into Eq.~\eqref{eq:calibrate_samples}, and some tedious calculations, we can obtain Theorem~\ref{thm:calibration_samples}.

\begin{proof}[Proof of Theorem~\ref{thm:calibration_samples}.]
By the \textbf{MedianOfMeans} method, we have
\begin{align}
\Pr\sbra{\abs{\hat{f}_{k}-f_{k}}\geq \epsilon_{f_k}}\leq \delta_c
\end{align}
with $N_c=34\max_k \Var{\hat{f}_k}/\epsilon_{f_k}^2$ and $K_c=2\ln (2\delta_c^{-1})$ for any error $\epsilon_{f_k}$ and failure probability $\delta_c$. With Lemma \ref{lem:calibrate_error_bound}, let $\epsilon_{f_k} = \frac{\abs{f_{2k}}\varepsilon_c}{1+ \varepsilon_c}$, then with $R_c=K_cN_c$ calibration samplings we have $|\tr{H_i\widehat{\Mcal}^{-1} \widetilde{\Mcal}(\rho)} - \tr{H_i\rho}|\leq \varepsilon_c \vabs{H_0}_{\infty}$ for any $i\in [m]$ with success probability $1-\delta_c$.
Hence, we have
\begin{align}
N_c &= 34 \max_{k}\Var{\hat{f}_k}\frac{(1+\varepsilon_c)^2}{\abs{f_{2k}}^2 \varepsilon_c^2}\\
&\leq  34 \frac{(1+\varepsilon_c)^2}{ \varepsilon_c^2}  \max_{k} \frac{\Ebb\sbra{\abs{\hat{f}_{2k}}^2}}{\abs{f_{2k}}^2}\\
&\leq 34 \frac{(1+\varepsilon_c)^2}{ \varepsilon_c^2} \max_k  \sum_{0\leq l \leq \min(k,n-k)}   \Bcal_{2l}
\frac{{n \choose l,k-l,l}_p^2 {2n \choose 2k}^{2}}{{2n \choose 2l,2k-2l, 2l}_p{n \choose k}^{4} \abs{\Bcal_k}^2}.
\end{align}
Hence the total number of samplings equals $R_c = K_cN_c = 68 \frac{(1+\varepsilon_c)^2\ln (2/\delta_c)}{ \varepsilon_c^2} \max_k  \sum_{0\leq l \leq \min(k,n-k)}   
\frac{\Bcal_{2l}{n \choose l,k-l,l}_p^2 {2n \choose 2k}^{2}}{{2n \choose 2l,2k-2l, 2l}_p{n \choose k}^{4} \abs{\Bcal_k}^2}$. Since $k \in \cbra{0,n}$ will give us the trivial values of $R_c = 68 \frac{(1+\varepsilon_c)^2\ln (2/\delta_c)}{ \varepsilon_c^2} \frac{\Bcal_0}{\abs{\Bcal_{k}}^2}$ where $k\in \cbra{0,n}$. 
In the following, we assume $k\in [n-1]$, and utilize the technique in concrete mathematics~\cite{graham1989concrete} to simplify this bound.
By the Stirling formula, with the assumption that $n\ne 0$, we have
\begin{align}
\frac{(2n)!}{(n!)^2} = \Theta\left(\frac{2^{2n}}{\sqrt{\pi n}}\right).
\label{eq:stirling_formula}
\end{align}
Hence as long as $k\not\in\cbra{0,n}$, we have
\begin{align}
\frac{{2n \choose 2k}}{{n \choose k}^{2}} 
&= \frac{(2n)!}{(n!)^2}\cdot 
 \frac{(n!)^2}{(2k)!}
 \cdot 
 \frac{(n-k)!^2}{(2n-2k)!}\\
 &= \Theta\pbra{\sqrt{ k\left(1-\frac{k}{n}\right)}}.
\end{align}
Similarly, if $l\not\in\cbra{0,k,n-k}$,  then we can derive
\begin{align}
\frac{{n \choose l,k-l,l}_p^2 }{{2n \choose 2l,2k-2l, 2l}_p} &= \Theta\pbra{\sqrt{\frac{n}{l^2(l^2-nl +k(n-k))}}}\\
&= \Theta\pbra{\frac{1}{l}\sqrt{\frac{1}{(k-\frac{k^2}{n}) - (l-\frac{l^2}{n})}}}.
\end{align}
Let
\begin{align}
g(l,k) &:= \frac{{2n \choose 2k}^2{n \choose l,k-l,l}_p^2}{{n \choose k}^{4}{2n \choose 2l,2k-2l, 2l}_p} \\
&= {\frac{c(k -\frac{k^2}{n})}{l\sqrt{(k -\frac{k^2}{n})-(l-\frac{l^2}{n})}}} \\
&= \frac{c\sqrt{n}(k -\frac{k^2}{n})}{l\sqrt{l^2-nl + k(n-k)}}
\end{align}
for some constant $c$, where $l\in[n]\backslash\cbra{k,n-k}$.
Let  $h(z) = \frac{c}{l\sqrt{n}}\sqrt{\frac{z^2}{l^2-ln-z}}$, then $h(z)$ is monotonically increasing for $z\in (0,l(n-l))$. Note that if we let $z$ be $k(n-k)$ in the function $h$, then we get the function $g(l,k)$. Hence for $0 < k < \frac{n}{2}$, the function $g(l,k)$ is monotonically increasing with respect to $k$, while for $\frac{n}{2} < k < n-1$, the function is monotonically decreasing with respect to $k$. For simplification, here we assume $k$ is even, the result also holds for odd $k$ with similar calculations.
Let $k = \frac{n}{2}$ we have 
\begin{align}
g\left(l,\frac{n}{2}\right) &= \frac{c\sqrt{n}\cdot n/4}{l(\frac{n}{2}-l)} \\
&=\frac{c\sqrt{n}}{2}\pbra{\frac{1}{l} + \frac{1}{\frac{n}{2}-l}}.
\end{align}
Hence
\begin{align}
\sum_{0< l< n/2} g(l,n/2) &\leq c\sqrt{n} \sum_{l\leq n/2} \frac{1}{l}\\
&<c\sqrt{n}\pbra{\ln (n/2) + \gamma +\frac{1}{n}}\\
&<c'\sqrt{n}\ln n
\end{align}
where $\gamma \approx 0.5772$ is the Euler–Mascheroni constant, and $c'$ is a constant.

In contrast, when the input $l$ is either 0, $k$, or $n-k$, it is straightforward to verify that $g(l,k) = \frac{{2n\choose 2k}}{{n\choose k}^2}$ and it equals $\Theta(\sqrt{k\pbra{1-\frac{k}{n}}}) = \Ord{\sqrt{n}}$ if $k\not\in \cbra{0,n}$.  When $k$ takes the values $0$ or $n$, the function $g(l,k)$ equals $1$. Therefore, in conclusion, we can obtain the upper bound of $\max_k \sum_{l} g(l,k)=\Ord{\sqrt{n}\ln n}$.

Hence the number of calibration samplings
\begin{align}
R_c =\Ord{\frac{\Bcal_{\max}\sqrt{n}\ln n \ln (1/\delta_c)}{\Bcal_{\min}^2\varepsilon_c^2}}.
\end{align}
In the assumption that $\Bcal_k$ are constants for any $k\in [n]$, it can be simplified to
\begin{align}
R_c = \Ord{\frac{\ln (1/\delta_c)\sqrt{n}\ln n}{\varepsilon_c^2}}.
\end{align}
\end{proof}
Since for physical noise such as depolarizing noise, {generalized} amplitude damping noise, and $X$-rotation noise, $\Bcal_k$ are constants for any $k\in [n]$, the number of calibration data equals $\Ord{\frac{\sqrt{n}\ln n\ln (1/\delta_c)}{\varepsilon_c^2}}$.

\section*{Supplementary note 8}
\label{app:variance_simplification}
In this section, we will analyze the explicit values for the variances given by Lemma \ref{lem:variance_estimation} in some specific instances. Note that $\abs{\tr{\widetilde{\gamma}_{S_3}\widetilde{\gamma}_{S_1}\rho}}\leq \abs{\tr{\rho}}\vabs{\widetilde{\gamma}_{S_3}\widetilde{\gamma}_{S_1}}_2\leq 1$. Hence for different observables, we have:
\begin{itemize}
    \item [(0)] Calculate $\tr{\rho \widetilde\gamma_S}$ for an unknown quantum state $\rho$ and observable $\widetilde\gamma_S$ where $\abs{S}=2k$, $c$ is a constant.
By Lemma \ref{lem:variance_estimation}, we have
\begin{align}
\Var{\hat v}&\leq \frac{(1 + \varepsilon_c)^2}{2^{2n}} \sum_{0\leq l_1 + l_2 + l_3\leq n} \sbra{l_1= 0}\sbra{l_2= k}\sbra{l_3= 0} \frac{(-1)^{l_1+l_2+l_3} {2n\choose 2l_2} \Bcal_{0}}{\Bcal_{l_2}^2{n\choose l_2}}\tr{\widetilde{\gamma}_{S}^2}\tr{\widetilde{\gamma}_{S}\widetilde{\gamma}_{S}^{\dagger}}\\
&=\frac{(1 + \varepsilon_c)^2{2n\choose 2k}\Bcal_{k}^{-2}}{{n\choose k}}\\
&=\Ocal(\Bcal_{k}^{-2}n^{k}).
\end{align}
Combined with the results in Theorem~\ref{thm:est_samples} and Theorem~\ref{thm:calibration_samples}, we can calculate all of $k$-RDMs with the number of estimation samplings 
\begin{align}
R_e = \Ocal\pbra{\frac{k\ln(n/\delta_e)}{\Bcal_{k}^{2}\varepsilon_e^2}n^{k}},
\end{align}
and the number of calibration samplings
\begin{align}
R_c = \Ord{\frac{\Bcal_{\max}\sqrt{n} \ln n \ln (1/\delta_c)}{\Bcal_{\min}^2\varepsilon_c^2}},
\label{eq:rc_supp2}
\end{align}
 the estimation error can be bounded to $\varepsilon_e+ \varepsilon_c$.
 \item[(1)] 
Calculate the overlap between an $n$-qubit Gaussian state with density matrices $\cbra{\rho_{g_j}}_{j=1}^m$ and any $n$-qubit quantum state with density matrix $\rho$, denoted as $\cbra{\tr{\rho {\rho_{g_j}}}}_{j=1}^m$. In the following, we will prove that by choosing the observable $H = {\rho_g}$, the number of required estimation samplings equals
\begin{align}
R_e = \Ocal\pbra{\frac{\sqrt{n}\Bcal_{\max}\ln(m/\delta_e)}{\Bcal_{\min}^2\varepsilon_e^2}}.
\end{align}
It can be further simplified to $R_e = \Ocal\pbra{\frac{\sqrt{n}\ln(m/\delta_e)}{\varepsilon_e^2}}$ for the noise with constant average noise fidelities $\cbra{\Bcal_k}$.

By the definition of Gaussian state $\rho_g$ in Eq.~\eqref{eq:gaussian_state_rep}, for any given sets $S_1,S_2\in [2n]$ with cardinalities $|S_1|=2l_1$, and $|S_2|=2l_2$, we have
\begin{align}
\tr{\widetilde{\gamma}_{S_1} \widetilde{\gamma}_{S_2} \rho_g} &= \frac{1}{2^{n}} (-1)^{l_1+l_2} \mu_{S} [S\subseteq[n],D(S) = S_1 \cup S_2],\\
&\leq \frac{1}{2^{n}} [S\subseteq[n], D(S) = S_1 \cup S_2]
\label{eq:mu_approx}
\end{align}
where $\mu_S = \prod_{j\in S} \mu_j$. Inequality \eqref{eq:mu_approx} holds since $\mu_S\in [-1,1]$. Therefore, 
\begin{align}
    \sum_{
\substack{S_1,S_2,S_3\ \mathrm{ disjoint}\\
|S_i| = 2l_i,i\in [3]}
}\tr{\widetilde\gamma_{S_1} \widetilde\gamma_{S_2}\rho_g} \tr{\widetilde\gamma_{S_2} \widetilde\gamma_{S_3} {\rho_g}^\dagger} \tr{\widetilde\gamma_{S_3}\widetilde\gamma_{S_1} \rho }\leq 1.
\end{align}
When combined with Lemma \ref{lem:variance_estimation}, we can bound the variance as follows:
 \begin{align}
\Var{\hat{v}} &\leq \sum_{0\leq l_1 + l_2 + l_3\leq n} \frac{{n\choose l_1,l_2,l_3}_p^2{2n \choose 2l_1+2l_2}{2n\choose 2l_2+2l_3}\abs{\Bcal_{l_1+l_3}}}{2^{2n}{2n \choose 2l_1,2l_2,2l_3}_p{n\choose l_1 + l_2}{n\choose l_2 + l_3}\abs{\Bcal_{l_1+l_2}\Bcal_{l_2+l_3}}}\\
&\leq \max_{0\leq l_1+l_2+l_3\leq n} \frac{{n\choose l_1,l_2,l_3}_p{2n \choose 2l_1+2l_2}{2n\choose 2l_2+2l_3}\abs{\Bcal_{l_1+l_3}}}{{2n \choose 2l_1,2l_2,2l_3}_p{n\choose l_1 + l_2}{n\choose l_2 + l_3}\abs{\Bcal_{l_1+l_2}\Bcal_{l_2+l_3}}}\sum_{0\leq l_1 + l_2 + l_3\leq n} \frac{{n\choose l_1,l_2,l_3}_p}{2^{2n}}\\
&\leq c \max_{1\leq l_1+l_2+l_3\leq n-1}\sqrt{\frac{n}{l_1l_2l_3(n-l_1-l_2-l_3)}}\cdot \sqrt{\frac{(l_1 + l_2)(n-l_1-l_2)}{n}} \cdot \sqrt{\frac{(l_2 + l_3)(n-l_2 - l_3)}{n}}\frac{\Bcal_{\max}}{\Bcal_{\min}^2}
\label{eq:simplify_choose}\\
& \leq c'\sqrt{n}\frac{\Bcal_{\max}}{\Bcal_{\min}^2},
\label{eq:simple_max_l}
 \end{align}
 where $\Bcal_{\max} = \max_k \abs{\Bcal_k}$ and $\Bcal_{\min} = \min_k \abs{\Bcal_k}$, $c,c'$ are constants. Eq.~\eqref{eq:simplify_choose} holds 
 since
\begin{align}
\frac{{2n \choose 2k}}{{n \choose k}^{2}} 
&= \frac{(2n)!}{(n!)^2}\cdot 
 \frac{(n!)^2}{(2k)!}
 \cdot 
 \frac{(n-k)!^2}{(2n-2k)!}\\
 &= \Theta\pbra{\sqrt{ k(1-\frac{k}{n})}},
\end{align}
if $k\not\in \cbra{0,n}$ by the Stirling formula~\cite{graham1989concrete}. Eq.~\eqref{eq:simple_max_l} holds by noting that the maximum value is obtained by choosing $l_1, l_3 $ in the order of $ \Theta\pbra{1}$, and $l_2 = \Theta\pbra{n}$. Hence the number of required estimation samplings 
\begin{align}
R_e = \Ocal\pbra{\frac{\sqrt n\Bcal_{\max}\ln(m/\delta_e)}{\Bcal_{\min}^2\varepsilon_e^2}}.
\end{align}
 \item[(2)] Calculate the overlap between $\tau$-fermionic Slater determinants $\cbra{\ket{\phi_{\tau_j}}}_{j=1}^m$ and any quantum pure state $\ket{\psi}$, denoted as $\cbra{\braket{\psi|\phi_{\tau_j}}}_{j=1}^m$. 
 Let the initial input state of the quantum device be $\rho = \frac{(\ket{\bm 0} + \ket{\psi})(\bra{\bm 0} + \bra{\psi})}{2}$, and observable $H = \ketbra{\phi_\tau}{\bm 0}$, the number of required estimation samplings
 \begin{align}
R_e = \Ocal\pbra{\frac{\sqrt{n}\ln n \Bcal_{\max}\ln(m/\delta_e)}{\Bcal_{\min}^2\varepsilon_e^2}},
 \end{align}
 which equals $\Ocal\pbra{\frac{\sqrt{n}\ln n \ln(m/\delta_e)}{\varepsilon_e^2}}$ if the average fidelities of noise $\Bcal_k$ are constants for any $k\in \cbra{0,\ldots, n}$.

 Let 
 \begin{align}
L := \max_{0\leq l_1+l_2+l_3\leq n} \frac{{n\choose l_1,l_2,l_3}_p{2n \choose 2l_1+2l_2}{2n\choose 2l_2+2l_3}\abs{\Bcal_{l_1+l_3}}}{{2n \choose 2l_1,2l_2,2l_3}_p{n\choose l_1 + l_2}{n\choose l_2 + l_3}\abs{\Bcal_{l_1+l_2}\Bcal_{l_2+l_3}}}\leq \frac{c\sqrt{n}\Bcal_{\max}}{\Bcal_{\min}^2}.
 \end{align}
 By Lemma \ref{lem:variance_estimation} and the properties of Slater determinant, we have
 \begin{align}
\Var{\hat v} &\leq \frac{L}{2^{2n}}\sum_{0\leq l_1+l_2+l_3\leq n}\sum_{
\substack{S_1,S_2,S_3\text{ disjoint}\\
|S_i| = 2l_i,i\in [3]}
} \tr{\widetilde\gamma_{S_1}\widetilde\gamma_{S_2}\ket{\phi_{\tau}}\bra{\bm 0} } \tr{\widetilde\gamma_{S_1}\widetilde\gamma_{S_2}\ket{\bm 0}\bra{\phi_{\tau}} }
\\&\leq  
\frac{L}{2^{2n}}\sum_{0\leq l_1 + l_2 + l_3\leq n}  \sum_{j = 0}^{\tau/2}{\tau \choose 2j}{n-\tau \choose l_1-\tau/2 + j, l_2-\tau/2 + j, l_3 -j, n-l_1 -l_2 -l_3-j}
\\
    & = \frac{L}{2^{2n}} \sum_{0\leq l_1 + l_2 + l_3\leq n} \sum_{j = J_{\min}}^{J_{\max}} {\tau \choose 2j}{n-\tau \choose l_1-\tau/2 + j, l_2-\tau/2 + j, l_3 -j, n-l_1 -l_2 -l_3-j}\\
& \leq \frac{L}{2^{2n}} \sum_{j = 0}^{\tau/2} {\tau \choose 2j}\sum_{0\leq l_1' + l_2' + l_3'\leq n-\tau} {n-\tau\choose l_1',l_2',l_3'}_p\\
& \leq \frac{L}{2^{2n}} \sum_{j = 0}^{\tau/2} 2^\tau 4^{n-\tau}
\label{eq:var_slater_deter_dev}\\
&\leq \frac{c \sqrt{n}\ln n\Bcal_{\max}}{\Bcal_{\min}^2}
 \end{align}
 where $J_{\min} = \max(\tau/2-l_1, \tau/2-l_2, l_3+\tau-n,l_1+l_2+l_3-n,0)$, and $J_{\max} = \max(n -\tau/2-l_1, n -\tau/2 - l_2, l_3,n-l_1-l_2-l_3,\tau/2)$ for some constant $c$, 
 where $\Bcal_{\max} = \max_k \abs{\Bcal_k}$ and $\Bcal_{\min} = \min_k \abs{\Bcal_k}$.
 Eq.~\eqref{eq:var_slater_deter_dev} holds since 
 \begin{align}
 \sum_{0\leq j_1+j_2+j_3\leq n-\tau}{n-\tau\choose j_1,j_2,j_3,n-\tau -j_1-j_2-j_3}= 4^{n-\tau}    
 \end{align}
 for any $j_1,j_2,j_3\in\cbra{0,\ldots, n-\tau}$, and $\sum_{j = 0}^{\tau/2} {\tau \choose 2j} \leq 2^\tau$.
 Hence the number of required estimation samplings
 \begin{align}
R_e = \Ocal\pbra{\frac{\sqrt{n}\ln n\Bcal_{\max}\ln(m/\delta_e)}{\Bcal_{\min}^2\varepsilon_e^2}}.
 \end{align}
\end{itemize}

\begin{figure}[t]
    \centering
    \includegraphics[width = 1.0\textwidth]{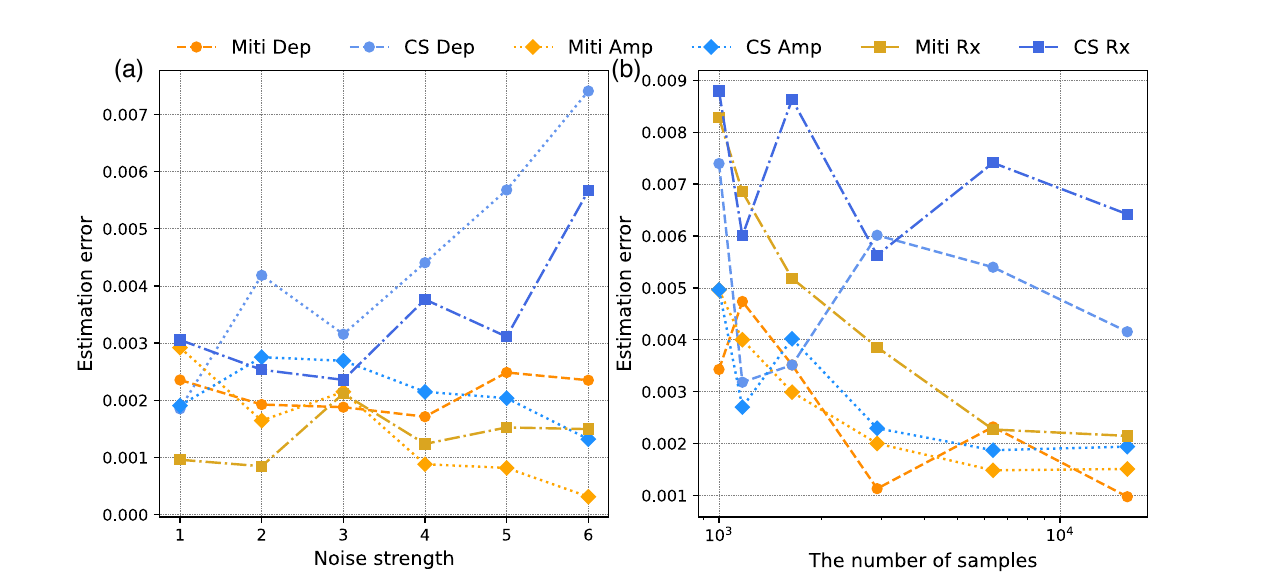}
    \caption{The estimation errors for the fidelities for the Gaussian state with an unknown quantum state change as (a) the increase of the noise strength and (b) the number of samples increases for a fixed noise parameter. The error $\varepsilon = \sqrt{\sum_{i = 1}^R (\hat{v}_i-\tr{\rho \widetilde{\gamma}_S})^2/R}$ of the estimation is obtained by repeating the procedure $R = 4$ rounds for (a) and $R=10$ for (b).
    }
    \label{fig:Gaussian}
\end{figure}

\section*{Supplementary note 9}
\label{app:numericals}
In this section, we give the numerical results for the calculation of (1) the overlap between a Gaussian state and any quantum state, and (2) the inner product between a Slater determinant and any pure state.

\begin{itemize}
    \item[(1)]{Mitigated estimation for the fidelities with fermionic Gaussian states}
    Let the observable $H $ be a $4$-qubit Gaussian state ${\rho_g} = \prod_{k = 1}^n\frac{(\Ibb - i\mu_k \widetilde{\gamma}_{2k-1}\widetilde{\gamma}_{2k})}{2}$, where $\cbra{\mu_k}_{k = 1}^n$ are uniformly randomly chosen from $[0,1]$.
We use the same noise settings for the expectation value of $H = {\rho_g}$ and the same number of calibration samplings as in the main file.
The numerical results are depicted in Fig.~\ref{fig:Gaussian}.  The number of estimation samplings 
 is set as $N = N_e K_e = 4000\times 5$ for matchgate CS, and $N_e = 4000/(1-p_{\text{noise}})$ for error-mitigated CS, as shown in Fig.~\ref{fig:Gaussian} (a), where $p_{\text{noise}}$ is the noise strength. Fig.~\ref{fig:Gaussian} (b) depicts the change of the estimation errors with the number of estimation samplings $N_e$ with the fixed noise parameters, where depolarizing noise parameter is as the fourth noise strength of Fig.~\ref{fig:Gaussian} (a), and the sixth  noise strength for both amplitude-damping and $X$-rotation noises. 
\item [(2)]{Mitigated estimation for overlaps with Slater determinants}

To numerically estimate the overlap between a pure state $\ket{\psi}$ and a $\tau$-Slater determinant $\ket{\phi_\tau}$, we uniformly randomly choose a normalized $3$-qubit pure state $\ket{\psi} = U_\psi\ket{0^3}$ and a unitary $U\in \Cbb^{3\times 3}$. $\tau$ is chosen as $1$. The initial state is generated as $\rho = \frac{(\ket{0^4}+\ket{1}\ket{\psi})(\bra{0^4}+\bra{1}\bra{\psi})}{2}$ with control-$U_\psi$ operation to the state $\ket{0^4}$. The observable  $H$ is chosen as $\ket{1}\ket{\phi_\tau}\bra{0^4} = \ket{1}\tilde{b}_1^\dagger\ket{0^3}\bra{0^4}$, where $\tilde{b}_1 = \sum_j V_{1j} b_j$ and unitary $V$ is randomly chosen.
We depict the numerical results in Fig.~\ref{fig:Slater}. The noise settings are the same as in Fig.~\ref{fig:Gaussian}. The number of estimation samplings 
 is set as $N = N_e K_e = 2000\times 5$ for matchgate CS, and $N_e = 2000/(1-p_{\text{noise}})$ for error-mitigated CS, as shown in Fig.~\ref{fig:Slater} (a), where $p_{\text{noise}}$ is the noise strength. The noise strength for 
Fig.~\ref{fig:Slater} (b) are the fourth, sixth, and sixth for depolarizing, {generalized} amplitude damping, and $X$-rotation noises respectively.
 \begin{figure}[t]
    \centering
    \includegraphics[width = 1.0\textwidth]{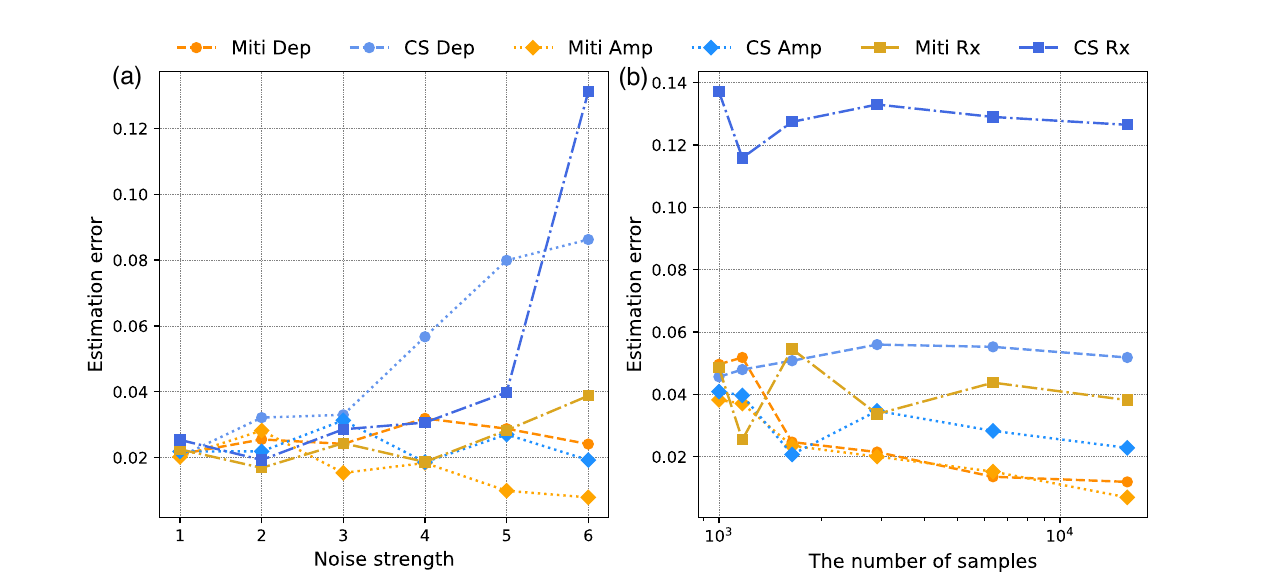}
    \caption{
    The estimation errors for the inner product of the $\tau$-Slater determinant and a pure state change as (a) the increase of the noise strength and (b) the number of samples increases for a fixed noise parameter. The error $\varepsilon = \sqrt{\sum_{i = 1}^R (\hat{v}_i-\tr{\rho \widetilde{\gamma}_S})^2/R}$ of the estimation is obtained by repeating the procedure $R = 4$ rounds.}
    \label{fig:Slater}
\end{figure}
\end{itemize}

\section*{Supplementary note 10}
\label{app:AveFidelity}

{
Average noise fidelity is proposed in randomized benchmarking \cite{Magesan2011scalable}, denoted as $F_{\text{avg}} =\int_{\psi} \braket{\psi |\Lambda(\ket{\psi}\bra{\psi}) |\psi}$. It is easy to check that $F_{\text{avg}} = 1 - (1-\frac{1}{2^n})p$ for depolarizing noise $\Lambda_{\text{d}}(A) = (1-p)(\cdot) + p\tr{A}\frac{\Ibb}{2^n}$. Since the expression of $F_{\text{avg}}$ for $\Lambda_{\text{a}}$ and $\Lambda_{\text{r}}$ is more intricate, we solely provide the explicit representation for a single-qubit system here. It has been demonstrated that the set $e =\{\ket{0},\ket{1},\ket{+},\ket{-},\ket{i},\ket{-i}\}$ constitutes a 2-design in the single-qubit state space~\cite{renes2004symmetric}. Therefore, for a single qubit system, $F_{\text{avg}} = \frac{1}{6}\sum_{\ket{\psi}\in e} \braket{\psi |\Lambda(\ket{\psi}\bra{\psi}) |\psi}$. This yields the values for $F_{\text{avg}}$ corresponding to the respective noises in Table 2 of the main file. With the expression values in Table 2 of the main file, we see that these quantities are closely aligned, and $\Bcal_1$ is slightly smaller than $F_{\text{avg}}$ and $F_Z$ for noise model $\Lambda_{\text{d}},\Lambda_{\text{a}}$. 

We plot the comparison of average fidelity $F_{\text{avg}}$, $Z$-basis average fidelity~\cite{chen2021robust} $F_Z$ and average fidelity in $\Gamma_2$ subspace $\Bcal_1$ for $X$-rotation noise for $\theta$ varies in $[0,2\pi]$ in Fig. \ref{fig:ave_fidelities}.
It shows that $\Bcal_1$ is much more subtle to noise.
}

\begin{figure}[htbp]
    \centering
    \includegraphics[width = 0.6\textwidth]{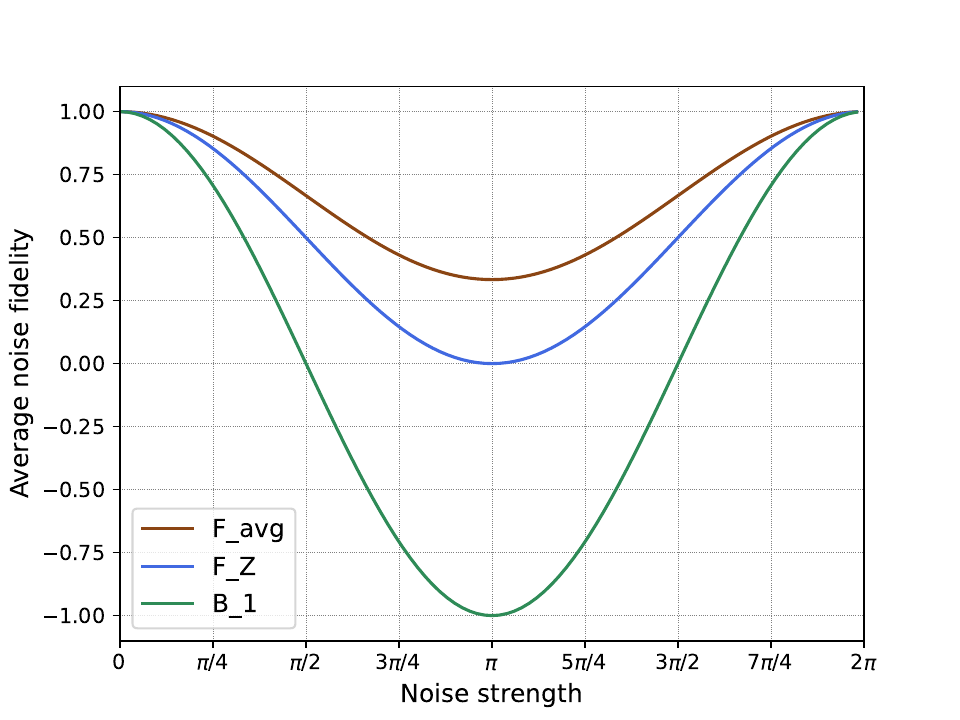}
    \caption{The average fidelities $F_{\text{avg}},F_Z$ and $\Bcal_1$ for $X$-rotation noise with noise range in $[0,2\pi]$.}
    \label{fig:ave_fidelities}
\end{figure}

\def\bibsection{}  

\centerline{ \textbf{SUPPLEMENTARY REFERENCES} }
\bigbreak

\end{document}